\documentclass{article}
\usepackage[utf8]{inputenc}

\usepackage{authblk}

\usepackage[hidelinks]{hyperref}

\usepackage{mathtools}

\usepackage{stmaryrd}

\usepackage{InqSem}

\usepackage{listings} 

\usepackage{complexity}

\usepackage{amsthm}
\newtheorem{PHONY}{PHONY}
\newtheorem{definition}[PHONY]{Definition}
\newtheorem{proposition}[PHONY]{Proposition}
\newtheorem{lemma}[PHONY]{Lemma}
\newtheorem{theorem}[PHONY]{Theorem}
\newtheorem{corollary}[PHONY]{Corollary}

\usepackage{tikz}
\tikzstyle{world}=[minimum width=0pt,inner sep=1pt,circle,draw=black,fill=none]

\newcommand{\tuple}[1]{\left\langle #1 \right\rangle}
\newcommand{\powset}{\mathcal{P}}
\newcommand{\card}[1]{\left|#1\right|}
\newcommand{\N}{\mathbb{N}}

\renewcommand{\th}{\text{th}}
\newcommand{\window}{\boxplus}

\newcommand{\APset}{\mathrm{AP}}

\newcommand{\model}[1]{\mathcal{#1}}

\newcommand{\bigO}{\mathcal{O}}
\newcommand{\bigTheta}{\Theta}
\newcommand{\size}[1]{\left| #1 \right|}
\newclass{\ATIME}{ATIME}
\newclass{\APclass}{AP}
\newclass{\ND}{ND}
\newlang{\TQBF}{TQBF}
\newlang{\MC}{MC}
\newlang{\MCanti}{\overline{MC}}

\newcommand{\PosSem}{\mathtt{PosSem}}
\newcommand{\NegSem}{\mathtt{NegSem}}

\newcommand{\SM}{\model{S}}

\newcommand\blfootnote[1]{%
  \begingroup
  \renewcommand\thefootnote{}\footnote{#1}%
  \addtocounter{footnote}{-1}%
  \endgroup
}

\title{Complexity of the Model Checking problem for inquisitive propositional and modal logic}
\author[1]{Gianluca Grilletti}
\affil[1]{Munich Center for Mathematical Philosophy}
\author[2]{Ivano Ciardelli}
\affil[2]{University of Padua}

\begin{document}

\maketitle

\blfootnote{Gefördert durch die Deutsche Forschungsgemeinschaft (DFG), Projektnummer 446711878. Funded by the German Research Foundation (DFG), project number 446711878.}

\section{Introduction}\label{section:introduction}

% Aim of the paper
The aim of this paper is to study the complexity of the \emph{model checking problem} $\MC$ for \emph{inquisitive propositional logic} $\inqB$ and for \emph{inquisitive modal logic} $\inqM$, that is, the problem of deciding whether a given finite structure for the logic satisfies a given formula.\footnote{In the terminology introduced by Vardi in \cite{Vardi:82}, this is the so-called \emph{combined complexity} for the model checking problem.}
In particular, we prove that both problems are $\APclass$-complete.
% and its modal extension \emph{inquisitive modal logic} $\inqM$.
% We prove that the problem is in the complexity class $\APclass$.
% Moreover, we present some expressive fragments \gnote{(only one?)} of these logics for which the corresponding $\MC$ problem has considerably lower complexities.

% Idea behind semantics of questions
\emph{Inquisitive logics} are a family of formalisms that extend classical and non-classical logic to encompass \emph{questions}.
Traditionally, logic is concerned with \emph{statements}, that is, expressions completely determined by their \emph{truth-conditions}.
The standard approaches to semantics are based on this assumption and are usually formulated in terms of truth-assignment (e.g., \emph{Boolean valuations} in propositional logic).
However, this semantical approach does not allow us to study expressions whose meaning is \emph{not} determined in terms of truth-conditions, such as \emph{questions}.
% Questions are not ``true'' or ``false'' in a given context, but rather they are requests for \emph{information on the context itself}.
% Nonetheless questions have a rich underlying logical structure, not expressed and explored by the truth-conditional semantics.

% inquisitive logics encompass this by using information
In order to interpret questions, which are not true or false, the semantics of \emph{Inquisitive logic} is based on a relation called \emph{support} between and \emph{information states} and a sentence.
This allows to give a uniform semantic account for statement and questions alike:
an information state \emph{supports} a statement if it \emph{implies} the statement;
and it \emph{supports} a question if it \emph{resolves} the question.
% To overcome the expressive limitations of the standard approaches, the semantics of \emph{Inquisitive logic} is based on \emph{information states} rather than \emph{truth assignments}.
% This allows to give a uniform semantic account for statement and questions alike:
% an information state \emph{supports} a statement if it \emph{implies} the statement;
% and it \emph{supports} a question if it \emph{resolves} the question.

% the two systems
The particular instances of inquisitive logic we consider in this paper are $\inqB$ and $\inqM$, extending propositional logic and modal logic respectively with the question-forming operators.
% ---and $\inqM$---extending standard modal logic with $\ivee$ and an additional modality $\window$.
In these systems information states are formalized as \emph{sets of truth-assignments}, which gives us a much richer semantical structure than in the classical setting.
The main question we tackle in this paper is:
how hard is it to \emph{decide if a formula of $\inqB$ (resp., $\inqM$) is supported by an information state in a given model}?
% , while in the latter they are formalized as \emph{sets of worlds}.
% In both cases, we obtain a much richer semantical structure than the classical counterparts of the logics, which leads to the main question of the paper:
% how hard is it to \emph{compute if a formula is supported by an information state}?

% results in dependence logic and connections with inquisitive logic
The problem of deciding whether a formula $\phi$ of a logic $L$ is satisfied by a model $\model{M}$ is referred to as the \emph{model checking problem for $L$} (in symbols $\MC(L)$).
The computational complexity of this problem is known for several logics (see, e.g., \cite{Clarke:01} for an overview of the classical results).
% and standard modal logic \gnote{(cite classical result)}.
In recent years, the problem has been addressed also 
for a class of logics called \emph{team logics} which, like our inquisitive systems, are interpreted relative to sets of assignments (see, e.g., \cite{Ebbing:12,Ebbing:13} and \cite[Ch.~7]{Yang:14}).
% \gnote{Here a brief overview of the computational results in team semantics. Fan's thesis + recent results from Hella.}.

% What we prove (a bit extended)
Although there is a close connection between these logics and inquisitive logics, it was still an open question what are the complexities of $\MC(\inqB)$ and $\MC(\inqM)$.\footnote{Both model checking problems have been investigated in \cite{Zeuner:20} (unpublished), building up on previous work by Yang \cite{Ebbing:13}.}
% An algorithm to solve the model checking problem for $\inqM$ problems has been presented and proved correct in \gnote{Yang}, showing that they belong to the complexity class $\AP$.
% An algorithm to solve both problems has been presented and proved correct in \gnote{Yang}, showing that they belong to the complexity class $\AP$.
In this paper we give a reduction of the $\PSPACE$-complete problem \emph{true quantified Boolean formulas} $\TQBF$ (see, e.g., \cite[Ch.~8]{Sipser:13}) to $\MC(\inqB)$, thus settling that both $\MC(\inqB)$ and $\MC(\inqM)$ are \emph{$\PSPACE$-complete}.
% Moreover, we study the model checking problems for the fragments with implication-depth $1$ of $\inqB$ and $\inqM$.
% Adapting the proof presented in \gnote{Yang}, we show that both problems have complexity $\ND$, thus exhibiting that the semantic of the implication is the main factor in the high complexity for the model checking problem.

% Structure of the paper
The paper is structured as follows:
In Section \ref{section:preliminaries} we present the basic notions of inquisitive logic used throughout the paper.
In Section \ref{section:algorithm} we introduce suitable encodings of $\inqB$ and $\inqM$ structures and formalize the model checking problems $\MC(\inqB)$ and $\MC(\inqM)$.
Moreover, in the same section we present an algorithm for \emph{alternating Turing machines} to solve $\MC(\inqM)$ and use it to show that $\MC(\inqM)$---and a fortiori also $\MC(\inqBQ)$---is in the class $\PSPACE$.
In Section \ref{section:complexity} we present and study a polynomial-space reduction of the $\TQBF$ problem to the $\MC(\inqB)$, thus showing that the problem is also $\PSPACE$-hard, and we thus infer that $\MC(\inqB)$ and $\MC(\inqM)$ are both $\PSPACE$-complete.
We conclude in Section \ref{section:conclusions} with some remarks and directions for future work.

\section{Preliminaries}\label{section:preliminaries}

In this section we recall some basic notions of inquisitive logic that we use throughout the paper.
As an extended introduction on the topic of inquisitive logic we suggest \cite{Ciardelli:16}.
As for computational complexity, we assume the reader to be already familiar with the basic notions on the topic (we refer to \cite{Sipser:13} for an introduction) and we limit ourselves to recall the main results used in the paper.

\subsection{Inquisitive propositional and modal logic}

Henceforth, we indicate with $\APset$ a fixed set of atomic propositions.

\begin{definition}[Syntax of $\inqB$ and $\inqM$]
    The set of formulas of $\inqM$---also called \emph{inquisitive modal formulas}---is defined by the following grammar:
    \begin{equation*}
        \phi \;::=\;  \bot  \;|\;  p  \;|\;  \phi\land\phi  \;|\;  \phi\ivee\phi  \;|\;  \phi\to\phi  \;|\;  \Box \phi  \;|\;  \window \phi
    \end{equation*}
    where $p \in \APset$.
    The set of \emph{inquisitive (propositional) formulas} $\inqB$ is the set of $\inqM$ formulas not containing the symbols $\Box$ and $\window$.
\end{definition}

\noindent
So the language of inquisitive logic is the standard language of modal logic extended with $\ivee$ and $\window$.
The operator $\ivee$ is called \emph{inquisitive disjunction}, while $\window$ is the \emph{window modality}.
Additionally, we introduce the following standard shorthands:
\begin{equation*}
    \neg \phi := \phi \to \bot
        \hspace{5em}
    \phi \vee \psi := \neg ( \neg \phi \land \neg \psi )
\end{equation*}

\noindent
The role of $\ivee$ is to enhance the language of modal logic with \emph{alternative questions}.
For example, the formula $p \ivee \neg p$ stands for the \emph{natural language question} ``does $p$ hold (or not)?''.
Instead, the derived operator $\vee$ corresponds to the usual disjunction from propositional logic:
for example, the expression $p \vee \neg p$ is intuitively interpreted as the \emph{sentence} ``$p$ holds or $p$ does not hold''.
% For a formal account of the different semantic interpretation of the two formulas, we point the reader to \gnote{cite}.
We will show the formal difference in the interpretation of these two formulas after introducing the semantics of the logic.

The operator $\window$ plays the role of an \emph{inquisitive modality}, sensitive to the inquisitive structures of the models.
For example, in the context of epistemic logic (see, e.g., \cite{CiardelliRoelofsen:15idel}) where the operator $\Box$ is interpreted as \emph{knowledge} of an abstract agent (e.g., $\Box p$ usually stands for ``the agent knows that $p$''), the operator $\window$ is roughly interpreted as a \emph{wondering} operator, encoding what the agent wonders about (e.g., $\window ?p$ stands for ``the agent wonders whether $p$ is the case'').

% We call formulas not containing the symbol $\ivee$ (i.e., formulas in the standard propositional language) \emph{classical} and we usually indicate these formulas with the symbols $\alpha$, $\beta$ and $\gamma$.
% As it can be easily verified, whenever $\alpha$ and $\beta$ are classical also $\neg \alpha$ and $\alpha \vee \beta$ are classical too.
As pointed out in Section \ref{section:introduction}, the presence of questions requires to move from a semantics based on \emph{truth-assignments} to one based on \emph{information states}.
To formalize this intuition, the logic employs special semantic structures, called \emph{information models}.

% Classical formulas are intuitively interpreted as \emph{statements}, in the same way as for propositional logic.
% For example, the formula $p \land q$ stands for the statement ``$p$ holds and $q$ holds''---or in short ``$p$ and $q$''.
% Inquisitive disjunction introduces \emph{disjunctive questions} into the picture, so that the formula $p \ivee q$ is intuitively interpreted as the question ``does $p$ hold or does $q$ hold?''.

% To formalize this intuitive interpretation of formulas, Ciardelli et al. \cite{} developed a novel semantics based on \emph{information} instead of \emph{truth-conditions}.
% For $X$ a set, we indicate with $\powset(X)$ the powerset of $X$.

\begin{definition}[Information models]\label{definition:information model}
    An \emph{information model for $\inqB$} is a tuple of the form $\model{M} = \tuple{ W, V }$ where $W$ is a non-empty set (the \emph{worlds} of the model) and $V:\APset \to \powset(W)$ is a function called the \emph{valuation} of the model.

    An \emph{information model for $\inqM$} is a tuple of the form $\model{M} = \tuple{ W, V, \Sigma }$ where $W$ and $V$ are as before, and $\Sigma : W \to \powset(\powset(W))$ is a function called an \emph{inquisitive state map}, satisfying the following conditions for every $w\in W$:
    \begin{enumerate}
        \item $\Sigma(w) \ne \emptyset$;
        \item If $s \in \Sigma(w)$ and $t \subseteq s$, then $t \in \Sigma$.
    \end{enumerate}
    We refer to the second condition as \emph{downward closure of $\Sigma(w)$}.
\end{definition}

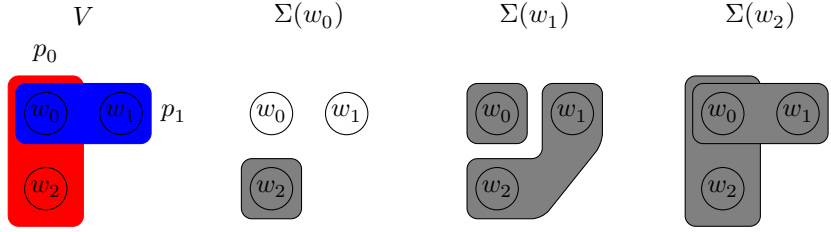
\begin{figure}
    \centering
    \begin{tikzpicture}
        
        \begin{scope}[shift={(0,0)}]
            \node at (.5,1.3) {$V$};

            \node at (0,.8) {$p_0$};
            \draw[info, red] (-.5,.5) rectangle (.5,-1.5);

            \node at (1.7,0) {$p_1$};
            \draw[info, blue] (-.4,.4) rectangle (1.4,-.4);

            \node[world] at (0, 0) {$w_0$};
            \node[world] at (1, 0) {$w_1$};
            \node[world] at (0,-1) {$w_2$};
        \end{scope}

        \begin{scope}[shift={(3,0)}]
            \node at (.5,1.3) {$\Sigma(w_0)$};

            \draw[info] (-.4,-.6) rectangle (.4,-1.4);

            \node[world] at (0, 0) {$w_0$};
            \node[world] at (1, 0) {$w_1$};
            \node[world] at (0,-1) {$w_2$};
        \end{scope}

        \begin{scope}[shift={(6,0)}]
            \node at (.5,1.3) {$\Sigma(w_1)$};

            \draw[info] (-.4,.4) rectangle (.4,-.4);
            \draw[info] (-.4,-.6) -- (.6,-.6) -- (.6,.4) -- (1.4,.4) -- (1.4,-.4) -- (.6,-1.4) -- (-.4,-1.4) -- cycle;

            \node[world] at (0, 0) {$w_0$};
            \node[world] at (1, 0) {$w_1$};
            \node[world] at (0,-1) {$w_2$};
        \end{scope}

        \begin{scope}[shift={(9,0)}]
            \node at (.5,1.3) {$\Sigma(w_2)$};

            \draw[info] (-.5,.5) rectangle (.5,-1.5);
            \draw[info] (-.4,.4) rectangle (1.4,-.4);

            \node[world] at (0, 0) {$w_0$};
            \node[world] at (1, 0) {$w_1$};
            \node[world] at (0,-1) {$w_2$};
        \end{scope}

    \end{tikzpicture}
    \caption{A graphical representation of an information model for $\inqM$ over the set of atoms $\APset = \{p_0, p_1\}$ and with set of worlds $W = \{w_0,w_1,w_2\}$.
    The valuation function $V$ is represented by the colored rectangles in the first image:
    $V(p_0) = \{ w_0, w_2\}$ (the red rectangle) and $V(p_1) = \{w_0,w_1\}$ (the blue rectangle).
    The map $\Sigma$ is represented in the other images by depicting the maximal states of each $\Sigma(w)$ for $w\in W$.
    So for example $\Sigma(w_2) = \{\, \{w_0,w_1\},\, \{w_0,w_2\} \,\}$, as depicted in the last image.
    Notice that to obtain a graphical representation of an information model for $\inqB$, we require only the first image of the sequence, that is, a representation of the valuation function $V$.
    }
    \label{fig:inqModel}
\end{figure}

\noindent
% Henceforth, for $\mathcal{S} \subseteq \powset(W)$ we define $\mathcal{S}^{\downarrow} := \{\,  Y\in \powset(W)  \,|\,  \exists X \in \mathcal{S}.\, Y \subseteq X  \,\}$.
% So the condition on $\Sigma$ from Definition \ref{definition:information model} can be reformulated as $\Sigma(w) = \Sigma(w)^{\downarrow}$ for every $w\in W$.
An example of information model for $\inqM$ is depicted in Figure \ref{fig:inqModel}.
Each world $w\in W$ is naturally associated with a \emph{Boolean valuation} over $\APset$, defined as
\begin{equation*}
    V_w(p) = \left\{ \begin{array}{ll}
        1  &\text{if } w \in V(p) \\
        0  &\text{otherwise}
    \end{array} \right.
\end{equation*}

\noindent
So we can think of each world as providing a \emph{complete description} of the current \emph{state of affairs}.
Under this interpretation, we can represent an information state as a \emph{set of worlds}, that is, all the worlds \emph{compatible} with the information state considered.
So for example, the information that ``$p$ holds'' can be represented as the set of worlds $w$ that assign truth value $1$ to $p$:
$\{\, w \in W \,|\, V_w(p) = 1 \,\} = V(p)$.
In line with this intuition, henceforth we refer to a subset $s \subseteq W$ as an \emph{information state} of the model.
The semantics of the logic is then defined relative to a model \emph{and} an information state.

% Intuitively, each $w \in W$ represents a \emph{possible world} or \emph{state of affairs} described by the propositional valuation
% \begin{equation*}
    % V_w(p) = \left\{ \begin{array}{ll}
        % 1  &\text{if } w \in V(p) \\
        % 0  &\text{otherwise}
    % \end{array} \right.
% \end{equation*}
% Thus we can encode a \emph{piece of information} as the set of those worlds compatible with it.
% So for example, the set encoding the information that ``$p$ has the same truth-value as $q$'' is $\{ w\in W \;|\; V_w(p) = V_w(q) \}$.
% For this reason, we call a subset $s \subseteq W$ an \emph{information state} of the model.
% Similarly for the case of possible-worlds semantics \cite{}, the semantics of an inquisitive formula is defined relative to am information model and an information state.

% \noindent
% So a natural way to interpret a classical formula $\alpha$ as an information state is to select the worlds satisfying $\alpha$, that is:\footnote{With a slight abuse of notation, we indicate with $V_w$ both the propositional valuation previously introduced and its natural extensions to propositional formulas.}
% \begin{equation*}
%     \left| \alpha \right|^{\model{M}} = \{\;  w \in W  \;|\;  V_w(\alpha) = 1  \;\}
% \end{equation*}

% \noindent
% To account for formulas of the extended language, we will define the semantics using a compositional approach, but later we will show that the semantics of a classical formula $\alpha$ at a model $\model{M}$ is fully determined by the set $\left| \alpha \right|^{\model{M}}$.

\begin{definition}[Semantics of $\inqM$]
    Let $\model{M}$ be an information and $s$ an information state of the model.
    We define the semantic relation $\vDash$ of $\inqM$ by the following inductive clauses:
    \begin{equation*}
    \begin{array}{lll}
        \model{M}, s \vDash \bot
            &\iff
            &s = \emptyset  \\
        \model{M}, s \vDash p
            &\iff
            &s \subseteq V(p)  \;\;\iff\;\;  \forall w\in s.\; V_w(p) = 1  \\
        \model{M}, s \vDash \phi \land \psi
            &\iff
            &\model{M}, s \vDash \phi \text{ and } \model{M}, s \vDash \psi  \\
        \model{M}, s \vDash \phi \ivee \psi
            &\iff
            &\model{M}, s \vDash \phi \text{ or } \model{M}, s \vDash \psi  \\
        \model{M}, s \vDash \phi \to \psi
            &\iff
            &\text{For all $t\subseteq s$, if } \model{M}, t \vDash \phi \text{ then } \model{M}, t \vDash \psi  \\
        \model{M}, s \vDash \Box \phi
            &\iff
            &\text{For all $w\in s$, } \model{M}, \bigcup\Sigma(w) \vDash \phi  \\
        \model{M}, s \vDash \window \phi
            &\iff
            &\text{For all $t\in \Sigma[s]$, } \model{M}, t \vDash \phi
    \end{array}
    \end{equation*}

    \noindent
    If $\model{M}, s \vDash \phi$ we say that $s$ \emph{supports} $\phi$.
\end{definition}

% \noindent
% In addition to the conditions above, we can derive support conditions for the operators $\neg$ and $\vee$ (see, e.g., \cite[Prop.~2.1.5]{Ciardelli:16}).

% \begin{equation*}
%     \begin{array}{lll}
%         \model{M}, s \vDash \neg \phi
%             &\iff
%             &\forall w \in s.\;  \model{M}, \{w\} \nvDash \phi  \\
%         \model{M}, s \vDash \phi \vee \psi
%             &\iff
%             &\forall w \in s.\;  \model{M}, \{w\} \vDash \phi \;\text{or}\;  \model{M}, \{w\} \vDash \psi
%     \end{array}
%     \end{equation*}

\noindent
Intuitively, an information state $s$ supports a formula $\phi$ if the information represented by $s$ \emph{implies the statement represented by $\phi$} or \emph{resolves the question represented by $\phi$}.
For example, consider the formula $p \ivee \neg p$ representing the question ``Does $p$ hold?''.
The support conditions for this formula are:
\begin{equation*}
\begin{array}{ll}
    \model{M}, s \vDash p \ivee \neg p
        &\iff   \model{M}, s \vDash p  \;\text{or}\;  \model{M}, s \vDash \neg p  \\
        &\iff  \left( \forall w \in s.\; V_w(p) = 1 \right)  \;\text{or}\;  \left( \forall w \in s.\; V_w(p) = 0 \right)
\end{array}
\end{equation*}

% For example, consider the formula $p \vee \neg p$ representing the statement ``$p$ holds or $p$ does not hold''---which is a classical tautology.
% The support conditions for this formula are:
% \begin{equation*}
% \begin{array}{ll}
%     \model{M}, s \vDash p \vee \neg p
%         &\iff   \forall w\in s.\; \left( \model{M}, \{w\} \vDash p  \;\text{or}\;  \model{M}, \{w\} \vDash \neg p \right) \\
%         &\iff   \forall w\in s.\; \left( V_w(p) = 1  \;\text{or}\;  V_w(p) = 0  \right)
% \end{array}
% \end{equation*}

% \noindent
% which is a tautological condition, as expected.

% Instead, the formula $p \ivee \neg p$, representing the \emph{question} ``Does $p$ hold?'' has the following support conditions:
% \begin{equation*}
% \begin{array}{ll}
%     \model{M}, s \vDash p \ivee \neg p
%         &\iff   \model{M}, s \vDash p  \;\text{or}\;  \model{M}, s \vDash \neg p  \\
%         &\iff  \left( \forall w \in s.\; V_w(p) = 1 \right)  \;\text{or}\;  \left( \forall w \in s.\; V_w(p) = 0 \right)
% \end{array}
% \end{equation*}

\noindent
So a state $s$ supports the formula $p \ivee \neg p$ if \emph{either} all the worlds in $s$ agree on $p$ being true \emph{or} all the worlds in $s$ agree on $p$ being false.
That is, if we have enough information to \emph{either} affirm that $p$ holds \emph{or} affirm that $p$ does not hold.

We conclude the Section with a proposition, exhibiting two fundamental properties of the semantics.

\begin{proposition}\label{proposition:empty and persistency}
    Let $\phi$ be a formula of the language and $\model{M}$ be a model.
    \begin{description}
        \item[(Empty state)] $\model{M}, \emptyset \vDash \phi$.
        \item[(Persistency)] For every $s \subseteq t \subseteq W$, if $\model{M}, s \vDash \phi$ then $\model{M}, t \vDash \phi$.
    \end{description}    
\end{proposition}

% \noindent
% We indicate the set of information states of $\model{M}$ that support the formula $\phi$ with the notation $\sem{\phi}^{\model{M}}$, that is $\sem{\phi}^{\model{M}} := \{  s \subseteq W  \;|\;  \model{M}, s \vDash \phi  \}$.
% The two properties above can be equivalently stated as a property of this set:
% For every formula $\phi$, $\sem{\phi}^{\model{M}} \subseteq \powset(\powset(W))$ is a non-empty collection closed under subsets.

% For classical formulas we can prove a stronger property.

% \begin{proposition}
%     Let $\alpha$ be a classical formula, $\model{M}$ be a model and $s$ be an information state of $\model{M}$.
%     Then:
%     \begin{equation*}
%         \model{M}, s \vDash \alpha
%             \quad\iff\quad
%             \text{For all $w\in s$, }
%             V_w(\alpha) = 1
%     \end{equation*}

%     \noindent
%     In particular, $\sem{\alpha}^{\model{M}} = \{  s\subseteq W  \;|\;  s \subseteq \left| \alpha \right|^{\model{M}}  \}$.
% \end{proposition}

% \gnote{Add properties needed for the results in following section.}

% \gnote{Add subsystems?}

% \subsection{Inquisitive modal logic}

% \gnote{To add.}

\section{Encoding and model checking algorithm}\label{section:algorithm}

\noindent
Henceforth, we limit ourselves to work with finite sets of atomic propositions and with finite models (i.e., with models with finitely many worlds).
% In particular, we assume throughout the section that $\APset = \{ p_0,\dots,p_{l-1} \}$ and $W = \{w_0, \dots, w_{n-1}\}$, and in particular $l$ and $n$ indicate the cardinalities of the sets $\APset$ and $W$ respectively.
Throughout the section, we will use the following ad-hoc notations for ease of presentation:
\begin{itemize}
    \item The vocabulary consists of the atomic propositions $\APset = \{ p_0, \dots, p_{l-1} \}$, thus $l$ is the number of atomic propositions;
    \item The set of worlds is $W = \{ w_0, \dots, w_{n-1} \}$, thus $n$ is the number of worlds in the model;
    \item The state map is $\Sigma(w_i) = \{ S_0^i, \dots, S_{k_i-1}^i \}$ for every $w_i \in W$;
    \item We define $m$ as the ``size'' of the function $\Sigma$, that is, the value
    \begin{equation*}
        m   \;=\;  |\Sigma(w_0)| + \dots + |\Sigma(w_{n-1})|
            \;=\;  k_0 + \dots + k_{n-1}
    \end{equation*}
\end{itemize}

% the numbers $l$ and $n$ respectively as the size of the set of atomic propositions $\APset$ and the size of the set of worlds $W$ of the model considered.
% Moreover we use the number $m$ as the ``size'' of the function $\Sigma$, that is, the value:
% \begin{equation*}
%     m = |\Sigma(w_0)| + \dots + |\Sigma(w_{n-1})|
% \end{equation*}
% value $\Sigma_{w\in W} | \Sigma(w)|$, the ``size'' of the function $\Sigma$.

% Consider an information model with set of worlds $W = \{w_0, \dots, w_{n-1}\}$ over the vocabulary $\APset = \{ p_0,\dots,p_{l-1} \}$ and with $\Sigma(w_i) = \{S^i_0,\dots,S^i_{k_i-1}\}$ for every $w_i \in W$.

\noindent
Consider an information model.
We define a natural way to encode this model by storing the propositional valuation in $n+1$ binary strings as follows:
\begin{itemize}
    \item The first string is $\delta$ of length $nl$.
    The $(li+j)^\th$ bit of the string is $\delta_{li+j} = V_{w_i}(p_j)$.
    This representation allows to check conditions of the form $w_i \in V(p_j)$ in time linear in $nl$.
    \item The other strings are $\epsilon_{0},\dots,\epsilon_{n-1}$.
    The string $\epsilon_k$ has length $(n+1)k_i+1$ and contains in order a digit $\mathtt{0}$, a binary representation of length $n$ of $S^i_0$, a digit $\mathtt{0}$, a binary representation of length $n$ of $S^i_1$\dots and ends with a digit $\mathtt{1}$.
    This representation allows to check conditions of the form $w_i \in \bigcup\Sigma[s]$ in time linear in $mn$.
\end{itemize}
An example of encoding is given in Figure \ref{fig:encoding}.
Notice that to encode an information model for $\inqB$ we do not need to store the strings $\epsilon_0,\dots,\epsilon_{n-1}$, but just the string $\delta$.

\begin{figure}
    \centering
    % \begin{minipage}{.3\textwidth}
    % \centering
    % \begin{tikzpicture}
        
    %     \node[world] at (0, 0) {$w_0$};
    %     \node[world] at (1, 0) {$w_1$};
    %     \node[world] at (0,-1) {$w_2$};

    %     \node at (0,.8) {$p_0$};
    %     \draw[info, red] (-.5,.5) rectangle (.5,-1.5);

    %     \node at (1.7,0) {$p_1$};
    %     \draw[info, blue] (-.4,.4) rectangle (1.4,-.4);

    % \end{tikzpicture}
    % \end{minipage}
    % \begin{minipage}{.5\textwidth}
    \begin{equation*}
    \begin{array}{c}
        \delta\hspace{1em}
            \underbrace{
                \begin{array}{ccc}
                    \mathtt{1}
                        &\mathtt{1}
                        &\mathtt{0}  \\
                    p_0
                        &p_1
                        &p_2
                \end{array}
            }_{w_0}
            \underbrace{
                \begin{array}{ccc}
                    \mathtt{0}
                        &\mathtt{1}
                        &\mathtt{0}  \\
                    p_0
                        &p_1
                        &p_2
                \end{array}
            }_{w_1}
            \underbrace{
                \begin{array}{ccc}
                    \mathtt{1}
                        &\mathtt{0}
                        &\mathtt{0}  \\
                    p_0
                        &p_1
                        &p_2
                \end{array}
            }_{w_2}  \\[3em]
        \epsilon_0\hspace{1em}
            \mathtt{0}\;
            \underbrace{
                \mathtt{0}
                \mathtt{0}
                \mathtt{1}
            }_{ \{w_2\} }\;
            \mathtt{1}
        \hspace{3em}
        \epsilon_1\hspace{1em}
            \mathtt{0}\;
            \underbrace{
                \mathtt{1}
                \mathtt{0}
                \mathtt{0}
            }_{ \{w_0\} }\;\,
            \mathtt{0}
            \underbrace{
                \mathtt{0}
                \mathtt{1}
                \mathtt{1}
            }_{ \{w_1,w_2\} }
            \mathtt{1}
        \hspace{3em}
        \epsilon_2\hspace{1em}
            \mathtt{0}
            \underbrace{
                \mathtt{1}
                \mathtt{1}
                \mathtt{0}
            }_{ \{w_0,w_1\} }
            \mathtt{0}
            \underbrace{
                \mathtt{1}
                \mathtt{0}
                \mathtt{1}
            }_{ \{w_0,w_2\} }
            \mathtt{1}
    \end{array}
    \end{equation*}
    % \end{minipage}
    \caption{The encoding of the information model depicted in Figure \ref{fig:inqModel}.
    Each bit of the encoding correspond to the valuation of a propositional atom at a certain world.
    For example, the bit $\mathtt{1}$ at position $0$ indicates that $w_0 \in V(p_0)$, while the bit $\mathtt{0}$ at position $2$ indicates that $w_0 \notin V(p_2)$.
    }
    \label{fig:encoding}
\end{figure}

We use a similar representation for \emph{information states}, encoded as strings of length $n$---where a $1$ in position $i$ indicates that the world $w_i$ is part of the state.
So, for example, the state $\{ w_0, w_2 \}$ in a model with worlds $\{w_0, w_1, w_2 \}$ is encoded by the binary string $\mathtt{101}$.
Given these encodings, conditions of the form $\model{M}, s \vDash p_j$, $w \in \bigcup \Sigma[s]$ and $t \in \Sigma[s]$ can be checked in time $\bigO((n(l+m))^2)$.\footnote{This and others bounds we indicate in this section are far from optimal, but this will not affect our analysis of the complexity of the model checking problem.}

% In case the values of $n$ and $l$ are not fixed (as it is the case for the model checking problem for $\inqB$), to encode a model we also need to store the value of $l$ in the string.
% Notice that this does not change the asymptotic bound $\bigO(N)$ on the size of the model, nor the bound in time to check the conditions $w_i \in V(p_j)$ and $\model{M}, s \vDash p_j$.

With a slight abuse of notation, we use the symbols $\model{M}$, $s$ and $\phi$ to indicate both the formal objects introduced and their binary encodings.
We define the \emph{model checking problem for $\inqM$} as the decision problem
\begin{equation*}
    \MC(\inqM) \;:=\; \{\;  \tuple{\model{M}, s, \phi} \;|\; \text{$\model{M}$ is a model for $\inqM$ and }  \model{M}, s \vDash \phi  \;\}
\end{equation*}

\noindent
We also introduce the corresponding anti-satisfaction problem
\begin{equation*}
    \MCanti(\inqM) \;:=\; \{\;  \tuple{\model{M}, s, \phi} \;|\; \text{$\model{M}$ is a model for $\inqM$ and }  \model{M}, s \nvDash \phi  \;\}
\end{equation*}

\noindent
Similarly, we introduce the corresponding problems for $\inqB$.
\begin{gather*}
    \MC(\inqB) \;:=\; \{\;  \tuple{\model{M}, s, \phi} \;|\; \text{$\model{M}$ is a model for $\inqB$ and }  \model{M}, s \vDash \phi  \;\}  \\
    \MCanti(\inqB) \;:=\; \{\;  \tuple{\model{M}, s, \phi} \;|\; \text{$\model{M}$ is a model for $\inqB$ and }  \model{M}, s \nvDash \phi  \;\}
\end{gather*}

\noindent
We are interested in studying the complexity of these decision problems, and in particular we claim that all the problems have complexity $\PSPACE$.
Notice that the problem $\MC(\inqM)$ (resp., $\MCanti(\inqM)$) is strictly more complex than the problem $\MC(\inqB)$ (resp., $\MCanti(\inqB)$), so for our claim we just need to prove that the former is a $\PSPACE$ problem and the latter is a $\PSPACE$-hard problem.

Firstly, we present two mutually recursive algorithms for alternating Turing machines (ATMs) (see, e.g., \cite[Ch.~10]{Sipser:13}) to solve $\MC(\inqM)$ and $\MCanti(\inqM)$:
$\PosSem$, which checks the condition $\model{M}, s \vDash \phi$, and $\NegSem$, which checks the condition $\model{M}, s \nvDash \phi$.
These algorithms are a variation of the algorithms presented in \cite[Ch.~7]{Yang:14} and \cite{Zeuner:20}.

% An analysis of these algorithms will allow us to provide an upper bound to the complexities of $\MC(\inqB)$ and $\MCanti(\inqB)$.

% Modulo a change in notation, the algorithm is essentially a variation of the algorithms presented in \gnote{Fan's thesis and Max's thesis}.
For brevity we will present the algorithms in pseudo-code.
We make use of the following expressions:
\begin{itemize}
    \item As usual, we indicate with \texttt{true}, \texttt{false} the Boolean values, and with \texttt{and}, \texttt{or}, \texttt{neg} the standard Boolean operators.

    \item Given a state $s$, we indicate with $\forall t \subseteq s$ that we are employing the universal states of the ATM to choose a subset $t \subseteq s$.
    In particular, the result yielded by the algorithm (relative to this non-deterministic branch of the computation) is \texttt{true} iff for all possible choices of $t$ the algorithm yields \texttt{true}.

    \item Given a state $s$, we indicate with $\exists t \subseteq s$ that we are employing the existential states of the ATM to choose a subset $t \subseteq s$.
    In particular, the result yielded by the algorithm (relative to this non-deterministic branch of the computation) is \texttt{true} iff for at least one choice of $t$ the algorithm yields \texttt{true}.
\end{itemize}

\noindent
Both operations can be implemented in an ATM by choosing non-deterministically the bits of the representation of $t$.
This amounts to $|s|$ non-deterministic choices, thus both operations have complexity in $\bigO(n)$.

\noindent
\begin{minipage}{\linewidth}
\begin{lstlisting}[
    escapeinside={(*}{*)},
    basicstyle=\footnotesize\ttfamily,
    % basicstyle=\scriptsize,
    breaklines=true,
    postbreak=\mbox{$\hookrightarrow$\space},
    numbers=left,
    xleftmargin=2em,
    framexleftmargin=1.5em
]
Algorithm PosSem
Input:  (*$\model{M}$*),(*$s$*),(*$\phi$*).
Output: true if (*$[\model{M},s \vDash \phi]$*), false otherwise.

Match (*$\phi$*) with
    (*$\phi = p_j$*):
        define result := true
        foreach i in {0,...,n-1}:
            if (*$w_i \in s$*) and (*$V_{w_i}(p_j) = 0$*) then result := false
        return result
    (*$\phi = \psi \land \chi$*):
        define result := PosSem((*$\model{M}$*),(*$s$*),(*$\psi$*)) and PosSem((*$\model{M}$*),(*$s$*),(*$\chi$*))
        return result
    (*$\phi = \psi \ivee \chi$*):
        define result := PosSem((*$\model{M}$*),(*$s$*),(*$\psi$*)) or PosSem((*$\model{M}$*),(*$s$*),(*$\chi$*))
        return result
    (*$\phi = \psi \to \chi$*):
        (*$\forall t \subseteq s$*)
        define result := NegSem((*$\model{M}$*),(*$t$*),(*$\psi$*)) or PosSem((*$\model{M}$*),(*$t$*),(*$\chi$*))
        return result
    (*$\phi = \Box \psi$*):
        define result := true
        foreach i in {0,...,n-1}:
            if (*$w_i \in s$*) then:
                define (*$t$*) := (*$\bigcup\Sigma(w_i)$*)
                if NegSem((*$\model{M}$*),(*$t$*),(*$\psi$*)) then result := false
        return result
    (*$\phi = \window \psi$*):
        define result := true
        foreach i in {0,...,n-1}:
            if (*$w_i \in s$*) then foreach (*$t \in \Sigma(w)$*):
                if NegSem((*$\model{M}$*),(*$t$*),(*$\psi$*)) then result := false
        return result
\end{lstlisting}
\end{minipage}

\noindent
\begin{minipage}{\linewidth}
\begin{lstlisting}[
    escapeinside={(*}{*)},
    basicstyle=\footnotesize\ttfamily,
    % basicstyle=\scriptsize,
    breaklines=true,
    postbreak=\mbox{$\hookrightarrow$\space},
    numbers=left,
    xleftmargin=2em,
    framexleftmargin=1.5em
]
Algorithm NegSem
Input:  (*$\model{M}$*),(*$s$*),(*$\phi$*).
Output: true if (*$[\model{M},s \nvDash \phi]$*), true otherwise.

Match (*$\phi$*) with
    (*$\phi = p_j$*):
        return neg(PosSem((*$\model{M}, s, p_j$*)))
    (*$\phi = \psi \land \chi$*):
        define result := NegSem((*$\model{M}$*),(*$s$*),(*$\psi$*)) or NegSem((*$\model{M}$*),(*$s$*),(*$\chi$*))
        return result
    (*$\phi = \psi \ivee \chi$*):
        define result := NegSem((*$\model{M}$*),(*$s$*),(*$\psi$*)) and NegSem((*$\model{M}$*),(*$s$*),(*$\chi$*))
        return result
    (*$\phi = \psi \to \chi$*):
        (*$\exists t \subseteq s$*)
        define result := PosSem((*$\model{M}$*),(*$t$*),(*$\psi$*)) and NegSem((*$\model{M}$*),(*$t$*),(*$\chi$*))
        return result
    (*$\phi = \Box \psi$*):
        define result := false
        foreach i in {0,...,n-1}:
            if (*$w_i \in s$*) then:
                define (*$t$*) := (*$\bigcup\Sigma(w_i)$*)
                if NegSem((*$\model{M}$*),(*$t$*),(*$\psi$*)) then result := true
        return result
    (*$\phi = \window \psi$*):
        define result := false
        foreach i in {0,...,n-1}:
            if (*$w_i \in s$*) then foreach (*$t \in \Sigma(w)$*):
                if NegSem((*$\model{M}$*),(*$t$*),(*$\psi$*)) then result := true
        return result
\end{lstlisting}
\end{minipage}

\noindent
Both algorithms work as intended, as proved in the following proposition.

\begin{proposition}
    Let $\model{M}$ be an information model, $s$ an information state of $\model{M}$ and $\phi$ an inquisitive formula.
    \begin{itemize}
        \item $\PosSem$ with inputs $\model{M},s,\phi$ returns as output \texttt{true} iff $\model{M},s \vDash \phi$.
        \item $\NegSem$ with inputs $\model{M},s,\phi$ returns as output \texttt{true} iff $\model{M},s \nvDash \phi$.
    \end{itemize}
\end{proposition}
\begin{proof}
    The proof proceeds by induction on the structure of the formula $\phi$.
    The cases for conjunction and disjunction are a direct translation of the semantic clauses (or their contrapositive conditions), so we comment only the cases for atomic proposition, implication and modalities.
    \begin{itemize}
        \item \textbf{Case $\phi = p_j$ in $\PosSem$:} The pseudo-code for this case consists of the following lines:
        \begin{lstlisting}[
            escapeinside={(*}{*)},
            basicstyle=\footnotesize\ttfamily,
            % basicstyle=\scriptsize,
            breaklines=true,
            postbreak=\mbox{$\hookrightarrow$\space},
            numbers=left,
            xleftmargin=2em,
            framexleftmargin=1.5em,
            firstnumber=7
        ]
define result := true
foreach i in {0,...,n-1}:
    if (*$w_i \in s$*) and (*$V_{w_i}(p_j) = 0$*) then result := false
return result
        \end{lstlisting}
        \noindent
        The variable \texttt{result} (the output of this portion of code) is initially assigned value \texttt{true} (line 7) and this value may be updated to \texttt{false} in the \texttt{foreach}-loop (lines 8-9).

        Firstly assume that $\model{M}, s \vDash p_j$, i.e., for every $w_i \in s$ we have $V_{w_i}(p_j) = 1$.
        In this case, the condition of the \texttt{if-then} expression at line 17 (\texttt{if} $w_i \in s$ \texttt{and} $V_{w_i}(p_j) = 0$) is never satisfied, thus the variable \texttt{result} maintains the value \texttt{true}, and the algorithm returns the correct value \texttt{true} (line 10).

        Secondly assume that $\model{M}, s \nvDash p_j$, i.e., there exists $w_i \in s$ for which $V_{w_i}(p_j) = 0$.
        In this case, the condition of the \texttt{if-then} expression at line 17 is satisfied for $w_i$, thus the variable \texttt{result} is assigned value \texttt{false}, and the algorithm returns the correct value \texttt{false} (line 10).

        \item \textbf{Case $\phi = p_j$ in $\NegSem$:} The pseudo-code for this case consists of the following lines:
        \begin{lstlisting}[
            escapeinside={(*}{*)},
            basicstyle=\footnotesize\ttfamily,
            % basicstyle=\scriptsize,
            breaklines=true,
            postbreak=\mbox{$\hookrightarrow$\space},
            numbers=left,
            xleftmargin=2em,
            framexleftmargin=1.5em,
            firstnumber=7
        ]
return neg(PosSem((*$\model{M}, s, p_j$*)))
        \end{lstlisting}
        \noindent
        In this case the algorithm returns the Boolean negation of the output of $\mathtt{\PosSem(}\model{M},s,p_j\mathtt{)}$.
        That is, the algorithm returns \texttt{true} if $\model{M},s\nvDash p_j$ and it returns \texttt{false} if $\model{M},s\vDash p_j$, which is the expected behavior.

        \item \textbf{Case $\phi = \psi \to \chi$ in $\PosSem$:} The pseudo-code for this case consists of the following lines:
        \begin{lstlisting}[
            escapeinside={(*}{*)},
            basicstyle=\footnotesize\ttfamily,
            % basicstyle=\scriptsize,
            breaklines=true,
            postbreak=\mbox{$\hookrightarrow$\space},
            numbers=left,
            xleftmargin=2em,
            framexleftmargin=1.5em,
            firstnumber=15
        ]
(*$\forall t\subseteq s$*)
define result := NegSem((*$\model{M}$*),(*$t$*),(*$\psi$*)) or PosSem((*$\model{M}$*),(*$t$*),(*$\chi$*))
return result
        \end{lstlisting}

        \noindent
        Translating the algorithm in mathematical expressions and applying the inductive hypothesis, this piece of code returns the Boolean value corresponding to the following expression:
        \begin{equation*}
            \forall t \subseteq s.\; \left[\;  \model{M},t \nvDash \psi  \;\text{or}\;  \model{M},t \vDash \chi   \;\right],
        \end{equation*}
        which is exactly the semantic clause for $\model{M},s\vDash \psi \to \chi$.

        \item \textbf{Case $\phi = \psi \to \chi$ in $\NegSem$:} The pseudo-code for this case consists of the following lines:
        \begin{lstlisting}[
            escapeinside={(*}{*)},
            basicstyle=\footnotesize\ttfamily,
            % basicstyle=\scriptsize,
            breaklines=true,
            postbreak=\mbox{$\hookrightarrow$\space},
            numbers=left,
            xleftmargin=2em,
            framexleftmargin=1.5em,
            firstnumber=15
        ]
(*$\exists t\subseteq s$*)
define result := PosSem((*$\model{M}$*),(*$t$*),(*$\psi$*)) and NegSem((*$\model{M}$*),(*$t$*),(*$\chi$*))
return result
        \end{lstlisting}

        \noindent
        Translating the algorithm in mathematical expressions and applying the inductive hypothesis, this piece of code returns the Boolean value corresponding to the following expression:
        \begin{equation*}
            \exists t \subseteq s.\; \left[\;  \model{M},t\vDash \psi  \;\text{and}\;  \model{M},t\nvDash \chi  \;\right] 
        \end{equation*}
        which is equivalent to the following expression, that is, the semantic condition corresponding to $\model{M},s\nvDash \psi \to \chi$:
        \begin{equation*}
            \neg\big(\;\;  \forall t \subseteq s.\;  \left[\;  \model{M},t\nvDash \psi  \;\text{or}\;  \model{M},t\vDash \chi  \;\right]  \;\;\big)
        \end{equation*}

        \item \textbf{Case $\phi = \Box \psi$ in $\PosSem$:} The pseudo-code for this case consists of the following lines:
        \begin{lstlisting}[
            escapeinside={(*}{*)},
            basicstyle=\footnotesize\ttfamily,
            % basicstyle=\scriptsize,
            breaklines=true,
            postbreak=\mbox{$\hookrightarrow$\space},
            numbers=left,
            xleftmargin=2em,
            framexleftmargin=1.5em,
            firstnumber=22
        ]
define result := true
foreach i in {0,...,n-1}:
    if (*$w_i \in s$*) then:
        define (*$t$*) := (*$\bigcup\Sigma(w_i)$*)
        if NegSem((*$\model{M}$*),(*$t$*),(*$\psi$*)) then result := false
return result
        \end{lstlisting}

        \noindent
        The variable $\mathtt{result}$ (the output of this portion of code) is initially assigned value $\mathtt{true}$ (line 22) and the value is updated to $\mathtt{false}$ in the $\mathtt{foreach}$-loop only if there exists $w_i \in s$ such that $\mathtt{NegSem(}\model{M}\mathtt{,}\bigcup\Sigma(w_i)\mathtt{,}\psi\mathtt{)}$ (lines 23-26).
        Applying the inductive hypothesis, the piece of code returns the Boolean value corresponding to the following expression:
        \begin{equation*}
            \forall w_i \in s.\; \model{M}, \bigcup\Sigma(w_i) \vDash \psi
        \end{equation*}
        which is the semantic condition corresponding to $\model{M}, s \vDash \Box\psi$.

        \item \textbf{Case $\phi = \Box \psi$ in $\NegSem$:} The pseudo-code for this case consists of the following lines:
        \begin{lstlisting}[
            escapeinside={(*}{*)},
            basicstyle=\footnotesize\ttfamily,
            % basicstyle=\scriptsize,
            breaklines=true,
            postbreak=\mbox{$\hookrightarrow$\space},
            numbers=left,
            xleftmargin=2em,
            framexleftmargin=1.5em,
            firstnumber=19
        ]
define result := true
foreach i in {0,...,n-1}:
    if (*$w_i \in s$*) then:
        define (*$t$*) := (*$\bigcup\Sigma(w_i)$*)
        if NegSem((*$\model{M}$*),(*$t$*),(*$\psi$*)) then result := false
return result
        \end{lstlisting}

        \noindent
        This portion of code is exactly the same as the portion for $\Box\psi$ in the algorithm $\mathtt{PosSem}$ (lines 22-27), but with the values assigned to the variable $\mathtt{result}$ reversed.
        Thus applying an analogous reasoning as in the previous point, the piece of code returns the Boolean value corresponding to the following expression:
        \begin{equation*}
            \exists w_i \in s.\; \model{M}, \bigcup\Sigma(w_i) \nvDash \psi
        \end{equation*}
        which is the semantic condition corresponding to $\model{M}, s \nvDash \Box\psi$.

        \item \textbf{Case $\phi = \window \psi$ in $\PosSem$:} The pseudo-code for this case consists of the following lines:
        \begin{lstlisting}[
            escapeinside={(*}{*)},
            basicstyle=\footnotesize\ttfamily,
            % basicstyle=\scriptsize,
            breaklines=true,
            postbreak=\mbox{$\hookrightarrow$\space},
            numbers=left,
            xleftmargin=2em,
            framexleftmargin=1.5em,
            firstnumber=29
        ]
define result := true
foreach i in {0,...,n-1}:
    if (*$w_i \in s$*) then foreach (*$t \in \Sigma(w)$*):
        if NegSem((*$\model{M}$*),(*$t$*),(*$\psi$*)) then result := false
return result
        \end{lstlisting}

        \noindent
        The variable \texttt{result} (the output of this portion of code) is initially assigned value \texttt{true} (line 29) and this value may be updated to \texttt{false} in the \texttt{foreach}-loop in case $\mathtt{NegSem(}\model{M}\mathtt{,}t\mathtt{,}\psi\mathtt{)}$ holds for some $t \in \bigcup\Sigma[s]$ (lines 30-32).
        Translating this condition in mathematical formulas, this corresponds to the following expression:
        \begin{equation*}
            \forall t \in \Sigma[s].\;   \model{M}, t \vDash \psi
        \end{equation*}
        which is the semantic condition corresponding to $\model{M}, s \vDash \window \psi$.

        \item \textbf{Case $\phi = \window \psi$ in $\NegSem$:} The pseudo-code for this case consists of the following lines:
        \begin{lstlisting}[
            escapeinside={(*}{*)},
            basicstyle=\footnotesize\ttfamily,
            % basicstyle=\scriptsize,
            breaklines=true,
            postbreak=\mbox{$\hookrightarrow$\space},
            numbers=left,
            xleftmargin=2em,
            framexleftmargin=1.5em,
            firstnumber=26
        ]
define result := true
foreach i in {0,...,n-1}:
    if (*$w_i \in s$*) then:
        define (*$t$*) := (*$\bigcup\Sigma(w_i)$*)
        if NegSem((*$\model{M}$*),(*$t$*),(*$\psi$*)) then result := false
return result
        \end{lstlisting}

        \noindent
        This portion of code is exactly the same as the portion for $\window\psi$ in the algorithm $\mathtt{PosSem}$ (lines 29-33), but with the values assigned to the variable $\mathtt{result}$ reversed.
        Thus applying an analogous reasoning as in the previous point, the piece of code returns the Boolean value corresponding to the following expression:
        \begin{equation*}
            \exists t\in \Sigma[s].\; \model{M},t\nvDash \psi
        \end{equation*}
        which is the semantic condition corresponding to $\model{M}, s \nvDash \window\psi$.

    \end{itemize}
\end{proof}

\noindent
We leave to the eager reader the task to show that both algorithms have complexity $\bigO((n(l+m))^2 \cdot \size{\phi})$, where $\size{\phi}$ denotes the size of the encoding of the formula $\phi$.
Since the input consists of an encoding of the model together with an encoding of the formula (plus overhead), the size of the input is linear in $t = n(l+m) + \size{\phi}$.
So the previous analysis shows that the problems $\MC(\inqM)$ and $\MCanti(\inqM)$ both lie in $\ATIME(t^3) \subseteq \PSPACE$.\footnote{Recall that $\PSPACE = \AP := \bigcup_{i \in \N} \ATIME(t^i)$ (see \cite[Th.~10.21]{Sipser:13}).}

\begin{theorem}\label{theorem:mc inqm and mcanti inqm are pspace}
    The problems $\MC(\inqM)$ and $\MCanti(\inqM)$ are in $\PSPACE$.
\end{theorem}

% In the next section we show that the problems $\MC(\inqB)$ and $\MCanti(\inqB)$ are $\PSPACE$-hard, thus settling that all the problems are $\PSPACE$-complete.

\section{Complexity of Model Checking for $\inqB$}\label{section:complexity}

\noindent
In this section we provide a polynomial reduction of the $\PSPACE$-complete problem $\TQBF$ (the set of \emph{true quantified Boolean formulas}) to $\MC(\inqB)$.
Combined with the results from the previous section, this shows that $\MC(\inqB)$ and $\MC(\inqM)$ are $\PSPACE$-complete problems.
Moreover, since the problem $\MC(\inqB)$ trivially reduces to $\MCanti(\inqB)$, this suffices to show that also $\MCanti(\inqB)$ and $\MCanti(\inqM)$ are also $\PSPACE$-complete.
Since we are working with the logic $\inqB$, we will work only with models for $\inqB$ (recall that we introduced the graphical representation of these models in Figure \ref{fig:inqModel}).

The section is divided in several parts.
Firstly we introduce the family of \emph{switching models} and show how to use them to encode Boolean valuations.\footnote{We remark here that we the terminology \emph{Boolean valuations} to indicate functions with codomain $\{0,1\}$, not to be confused with the valuation of an information model.}
Secondly, we introduce some special formulas, later used to encode an instance of the $\TQBF$ problem (i.e., a quantified Boolean sentence) into an instance of $\MC(\inqB)$.
These formulas allow to simulate the steps of a non-deterministic search to solve the $\TQBF$ problem.
Thirdly, we employ the formulas previously defined to define the encoding and show that it is a reduction of the $\TQBF$ problem into the $\MC(\inqB)$ problem.
Finally, we show that the encoding proposed is polinomially bounded in size, thus proving that it preserves $\PSPACE$-hardness.

% \medskip
% \noindent
% \textbf{Switching models.}
\subsection{Switching models}
We introduce \emph{switching models}, a family of inquisitive models.
These models allows us to encode Boolean valuations as information states.
These models were originally employed in \cite{Ebbing:12} and later in \cite[Ch.~7]{Yang:14} to study the model checking problem for different team semantics.

\begin{definition}[Switching model] Given $l \in \N$ a positive natural number, we define the model $\SM_l = \tuple{W_l, V_l}$ over the set of propositional atoms $\{p_0,\dots,p_{l-1},q_0,\dots,q_{l-1}\}$ by the following clauses:
\begin{itemize}
    \item $W_l = \{ w_0^+, w_0^-, w_1^+, w_1^-, \dots, w_{l-1}^+, w_{l-1}^- \}$;
    \item $V_l(p_i) = \{ w_i^+ \}$;
    \item $V_l(q_i) = \{ w_i^+, w_i^- \}$.
\end{itemize}
\end{definition}

\begin{figure}
    \centering
    \begin{tikzpicture}
        
        \draw[info, green!50!black] (-.6,.6) node[above]{$p_1$} rectangle (.6,-.6);
        \draw[info, green!65!black] ( .9,.6) node[above]{$p_2$} rectangle (2.1,-.6);
        \draw[info, green!80!black] (3.9,.6) node[above]{$p_3$} rectangle (5.1,-.6);

        \draw[info, blue!100!white] (-.5,.5) rectangle ( .5,-1.5) node[below]{$q_1$};
        \draw[info,  blue!80!white] (  1,.5) rectangle (  2,-1.5) node[below]{$q_2$};
        \draw[info,  blue!60!white] (  4,.5) rectangle (  5,-1.5) node[below]{$q_3$};

        \node[world] (w1p) at (  0, 0) {$w_1^+$};
        \node[world] (w1m) at (  0,-1) {$w_1^-$};
        \node[world] (w2p) at (1.5, 0) {$w_2^+$};
        \node[world] (w2m) at (1.5,-1) {$w_2^-$};
        \node at (3,-.5) {$\cdots$};
        \node[world] (wlp) at (4.5, 0) {$w_l^+$};
        \node[world] (wlm) at (4.5,-1) {$w_l^-$};

    \end{tikzpicture}
    \caption{The switching model $\SM_l$.
    The circles represent the worlds of the model ($w_1^+,w_1^-,\dots$).
    The extension of the atomic propositions is depicted by the colored rectangles (in green the atoms $p_i$, in blue the atoms $q_i$).
    For example, the set $V_l(p_2) = \{ w_2^+ \}$ corresponds to the second (from the left) green rectangle.}
    \label{fig:exampleSwitchingModel}
\end{figure}
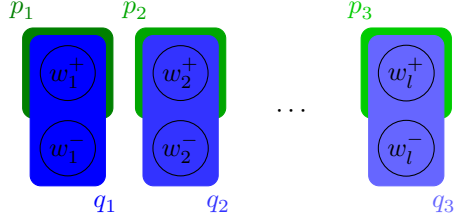

\noindent
A representation of a switching model is depicted in Figure \ref{fig:exampleSwitchingModel}.
For a fixed switching model $\SM_l$, we define special states to encode Boolean valuations over the set of atoms $\APset_k := \{x_0,\dots,x_{k-1}\}$ for $k\leq l$:
we will call these states \emph{$k$-switchings}.\footnote{Notice that we are using two distinct families of propositional atoms:
the $p_i$s and $q_i$s are the atoms evaluated on the switching model, while the $x_i$s (the elements of $\APset_k$) are the atoms used to define the Boolean valuations encode by information states.}

\begin{definition}\label{definition:switching}
    Consider a number $k \leq l$.
    A \emph{$k$-switching} of $\SM_l$ is an information state $s$ with the following two properties:
    \begin{itemize}
        \item For every $i < k$, $s$ contains exactly one world among $w_i^+$ and $w_i^-$;
        \item For every $i$ with $k \leq i < l$, $s$ contains both $w_i^+$ and $w_i^-$.
    \end{itemize}

\end{definition}

\noindent
Two examples of switchings are depicted in Figure \ref{fig:exampleSwitching}.
Notice that there exists only one $0$-switching, that is, the info state $W_l$.
There is a natural correspondence between $k$-switchings of $\SM_l$ and Boolean valuations over $\APset_k$.
Given a \emph{Boolean valuation} $\sigma: \APset_k \to \{0,1\}$, we define the corresponding $k$-switching $s_\sigma$ by the following clauses (for $i < k$):
\begin{equation*}
\begin{array}{lcl}
    w_i^+ \in s_\sigma  &\quad\iff\quad  &\sigma(x_i) = 1  \\
    w_i^- \in s_\sigma  &\quad\iff\quad  &\sigma(x_i) = 0
\end{array}
\end{equation*}

\noindent
This correspondence also holds for the particular case of $k=0$:
the set $\APset_0$ is empty, so there is only one valuation over this set (the empty function);
and the only $0$-switching of $\SM_l$ is $W_l$, the set of all worlds of the model.

Let us also highlight the following properties, which will come in handy to later define the encoding of a quantified Boolean formula.
For $i \leq k$ we have:
\begin{equation*}
\begin{array}{lcl}
    \sigma(x_i) = 1  &\quad\iff\quad  &s_\sigma \vDash q_i \to p_i  \\
    \sigma(x_i) = 0  &\quad\iff\quad  &s_\sigma \vDash q_i \to \neg p_i
\end{array}
\end{equation*}

\noindent
This map from Boolean valuations $\sigma$ to $k$-switchings $s_\sigma$ is bijective.
Its inverse $s \mapsto \sigma_s$ is defined by the following two equivalent conditions:
\begin{equation*}
\begin{array}{lcl}
    \sigma_s(x_i) = 1  &\quad\iff\quad  &w_i^+ \in s  \\
    \sigma_s(x_i) = 0  &\quad\iff\quad  &w_i^- \in s
\end{array}
\end{equation*}

\begin{figure}
    \centering
    \begin{tikzpicture}
        \draw[info, red!70!black] (-.5,.5) -- (-.5,-.5) -- (.5,-.5) -- (1,-1.5) -- (2,-1.5) -- (2.5,-.5) -- (5,-.5) -- (5,.5) -- (2.5,.5) -- (2,-.5) -- (1,-.5) -- (.5,.5) -- cycle;
        
        \node[world] (w1p) at (  0, 0) {$w_0^+$};
        \node[world] (w1m) at (  0,-1) {$w_0^-$};
        \node[world] (w2p) at (1.5, 0) {$w_1^+$};
        \node[world] (w2m) at (1.5,-1) {$w_1^-$};
        \node[world] (w3p) at (  3, 0) {$w_2^+$};
        \node[world] (w3m) at (  3,-1) {$w_2^-$};
        \node[world] (w4p) at (4.5, 0) {$w_3^+$};
        \node[world] (w4m) at (4.5,-1) {$w_3^-$};

        % \draw[info, green!50!black] (-.6,.6) node[above]{$p_1$} rectangle (.6,-.6);
        % \draw[info, blue!100!white] (-.5,.5) rectangle ( .5,-1.5) node[below]{$q_1$};
        % \draw[info, green!65!black] ( .9,.6) node[above]{$p_2$} rectangle (2.1,-.6);
        % \draw[info,  blue!80!white] (  1,.5) rectangle (  2,-1.5) node[below]{$q_2$};
        % \draw[info, green!80!black] (3.9,.6) node[above]{$p_3$} rectangle (5.1,-.6);
        % \draw[info,  blue!60!white] (  4,.5) rectangle (  5,-1.5) node[below]{$q_3$};

    \end{tikzpicture}\hspace{2em}
    \begin{tikzpicture}
        \draw[info, red!70!black] (-.5,.5) -- (-.5,-.5) -- (.5,-.5) -- (1,-1.5) -- (5,-1.5) -- (5,.5) -- (2.5,.5) -- (2,-.5) -- (1,-.5) -- (.5,.5) -- cycle;
        
        \node[world] (w1p) at (  0, 0) {$w_0^+$};
        \node[world] (w1m) at (  0,-1) {$w_0^-$};
        \node[world] (w2p) at (1.5, 0) {$w_1^+$};
        \node[world] (w2m) at (1.5,-1) {$w_1^-$};
        \node[world] (w3p) at (  3, 0) {$w_2^+$};
        \node[world] (w3m) at (  3,-1) {$w_2^-$};
        \node[world] (w4p) at (4.5, 0) {$w_3^+$};
        \node[world] (w4m) at (4.5,-1) {$w_3^-$};

        % \draw[info, green!50!black] (-.6,.6) node[above]{$p_1$} rectangle (.6,-.6);
        % \draw[info, blue!100!white] (-.5,.5) rectangle ( .5,-1.5) node[below]{$q_1$};
        % \draw[info, green!65!black] ( .9,.6) node[above]{$p_2$} rectangle (2.1,-.6);
        % \draw[info,  blue!80!white] (  1,.5) rectangle (  2,-1.5) node[below]{$q_2$};
        % \draw[info, green!80!black] (3.9,.6) node[above]{$p_3$} rectangle (5.1,-.6);
        % \draw[info,  blue!60!white] (  4,.5) rectangle (  5,-1.5) node[below]{$q_3$};

    \end{tikzpicture}
    \caption{A $4$-switching $s$ (on the left) and a $2$-switching $t$ (on the right) over the model $\SM_4$.
    $s$ corresponds to the Boolean valuation $\sigma: \APset_4 \to \{0,1\}$ defined as $\sigma(x_0) = 1$, $\sigma(x_1) = 0$, $\sigma(x_2) = 1$ and $\sigma(x_3) = 1$.
    $t$ corresponds to the Boolean valuation $\tau: \APset_2 \to \{0,1\}$ defined as $\tau(x_0) = 1$ and $\tau(x_1) = 0$.
    }
    \label{fig:exampleSwitching}
\end{figure}
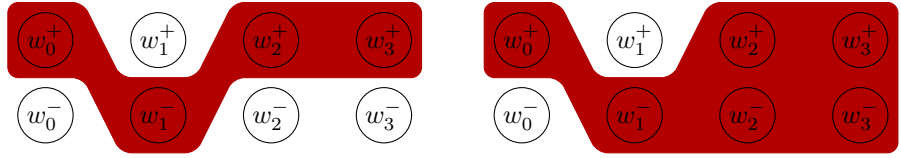

% \medskip
% \noindent
% \textbf{Some special formulas.}
\subsection{Some special formulas}
\label{subsection:specialFormulas}
We now introduce several formulas that will help us in our endeavors.
Each of these formulas has a characteristic semantics when interpreted on the switching model $\SM_l$.

\begin{equation}
    C_k^+ := q_k \land p_k
    \quad\text{and}\quad
    C_k^- := q_k \land \neg p_k
    \qquad \text{for $0\leq k < l$}
\end{equation}
These two formulas characterize the singleton states $\{w_k^+\}$ and $\{w_k^-\}$ respectively.
A depiction of their semantics is given in Figure \ref{fig:cFormulas}.
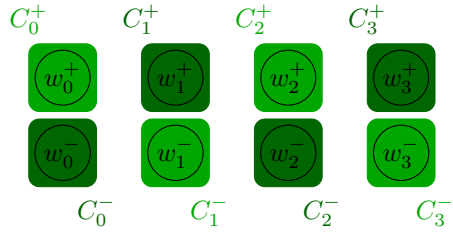
\begin{figure}
\centering
\begin{tikzpicture}
        \draw[info, green!65!black] (-.45, .45) node[above]{$C_0^+$} rectangle (.45,- .45);
        \draw[info, green!40!black] (-.45,-.55) rectangle (.45,-1.45) node[below]{$C_0^-$};
        
        \draw[info, green!40!black] (1.05, .45) node[above]{$C_1^+$} rectangle (1.95,- .45);
        \draw[info, green!65!black] (1.05,-.55) rectangle (1.95,-1.45) node[below]{$C_1^-$};

        \draw[info, green!65!black] (2.55, .45) node[above]{$C_2^+$} rectangle (3.45,- .45);
        \draw[info, green!40!black] (2.55,-.55) rectangle (3.45,-1.45) node[below]{$C_2^-$};

        \draw[info, green!40!black] (4.05, .45) node[above]{$C_3^+$} rectangle (4.95,- .45);
        \draw[info, green!65!black] (4.05,-.55) rectangle (4.95,-1.45) node[below]{$C_3^-$};
        
        \node[world] (w1p) at (  0, 0) {$w_0^+$};
        \node[world] (w1m) at (  0,-1) {$w_0^-$};
        \node[world] (w2p) at (1.5, 0) {$w_1^+$};
        \node[world] (w2m) at (1.5,-1) {$w_1^-$};
        \node[world] (w3p) at (  3, 0) {$w_2^+$};
        \node[world] (w3m) at (  3,-1) {$w_2^-$};
        \node[world] (w4p) at (4.5, 0) {$w_3^+$};
        \node[world] (w4m) at (4.5,-1) {$w_3^-$};

        % \draw[info, green!50!black] (-.6,.6) node[above]{$p_1$} rectangle (.6,-.6);
        % \draw[info, blue!100!white] (-.5,.5) rectangle ( .5,-1.5) node[below]{$q_1$};
        % \draw[info, green!65!black] ( .9,.6) node[above]{$p_2$} rectangle (2.1,-.6);
        % \draw[info,  blue!80!white] (  1,.5) rectangle (  2,-1.5) node[below]{$q_2$};
        % \draw[info, green!80!black] (3.9,.6) node[above]{$p_3$} rectangle (5.1,-.6);
        % \draw[info,  blue!60!white] (  4,.5) rectangle (  5,-1.5) node[below]{$q_3$};
    \end{tikzpicture}
    \caption{Semantics of $C_k^+$ and $C_k^-$.}
    \label{fig:cFormulas}
\end{figure}

\begin{equation}
    D_k := q_k \to {\,?p_k}
    \qquad\text{for $0\leq k < l$}
\end{equation}
The maximal states satisfying $D_k$ are depicted in Figure \ref{fig:Dformula}.
% These formulas allow us to characterize semantically the $1$-switchings of the model.
We give a short proof that this is indeed their semantics.

\begin{lemma}\label{lemma:dFormula}
    Let $s$ be a state of $\SM_l$.
    Then the maximal substates of $s$ satisfying $D_k$ are $t^+ = s_\sigma \setminus\{ w_k^- \}$ and $t^- = s_\sigma \setminus \{ w_k^+ \}$.
\end{lemma}

% \noindent
% Notice that $t^+$ and $t^-$ may be the same state, depending on the set $s$ considered.

\begin{proof}
    The only state of $\SM_l$ satisfying $q_k$ and not satisfying $?p_k$ is $\{w^+_k,w^-_k\}$.
    So by definition of $D_k$ for every $t \subseteq s$ we have that $\SM_l, t \vDash D_k$ iff $\{w^+_k,w^-_k\} \nsubseteq t$.
    In particular, $t^+$ and $t^-$ are the maximal substates of $s$ with these property.
\end{proof}

% \begin{lemma}
%     Let $k\leq l$ be two nonnegative integers.
%     The maximal states of $\SM_l$ satisfying $D_1\land D_2 \land \dots \land D_k$ are exactly the $k$-switchings of $\SM_l$.
% \end{lemma}
% \begin{proof}
%     A state $s$ of $\SM_l$ satisfies $D_i$ iff $\{ w^+_i,w^-_i \} \nsubseteq s$.
%     So we have for that
%     \begin{equation*}
%         \SM_l, s \vDash D_1 \land \dots \land D_k
%         \;\iff\;
%         \forall i \leq k.\; \{ w^+_i,w^-_i \} \nsubseteq s
%     \end{equation*}
%     It is trivial to verify that the maximal states with this property are exactly the $k$-switchings of $\SM_l$.
% \end{proof}

\begin{figure}
\centering
\begin{tikzpicture}
        \draw[info, red!70!black] (-.5,.5) -- (-.5,-1.5) -- (5,-1.5) -- (5,.5) -- (2.5,.5) -- (2,-.5) -- (1,-.5) -- (.5,.5) -- cycle;

        \node[world] (w1p) at (  0, 0) {$w_0^+$};
        \node[world] (w1m) at (  0,-1) {$w_0^-$};
        \node[world] (w2p) at (1.5, 0) {$w_1^+$};
        \node[world] (w2m) at (1.5,-1) {$w_1^-$};
        \node[world] (w3p) at (  3, 0) {$w_2^+$};
        \node[world] (w3m) at (  3,-1) {$w_2^-$};
        \node[world] (w4p) at (4.5, 0) {$w_3^+$};
        \node[world] (w4m) at (4.5,-1) {$w_3^-$};

        % \draw[info, green!50!black] (-.6,.6) node[above]{$p_1$} rectangle (.6,-.6);
        % \draw[info, blue!100!white] (-.5,.5) rectangle ( .5,-1.5) node[below]{$q_1$};
        % \draw[info, green!65!black] ( .9,.6) node[above]{$p_2$} rectangle (2.1,-.6);
        % \draw[info,  blue!80!white] (  1,.5) rectangle (  2,-1.5) node[below]{$q_2$};
        % \draw[info, green!80!black] (3.9,.6) node[above]{$p_3$} rectangle (5.1,-.6);
        % \draw[info,  blue!60!white] (  4,.5) rectangle (  5,-1.5) node[below]{$q_3$};
        
    \end{tikzpicture}\hspace{2em}
    \begin{tikzpicture}
        \draw[info, blue!70!black] (-.5,.5) -- (-.5,-1.5) -- (.5,-1.5) -- (1,-.5) -- (2,-.5) -- (2.5,-1.5) -- (5,-1.5) -- (5,.5) -- cycle;
        
        \node[world] (w1p) at (  0, 0) {$w_0^+$};
        \node[world] (w1m) at (  0,-1) {$w_0^-$};
        \node[world] (w2p) at (1.5, 0) {$w_1^+$};
        \node[world] (w2m) at (1.5,-1) {$w_1^-$};
        \node[world] (w3p) at (  3, 0) {$w_2^+$};
        \node[world] (w3m) at (  3,-1) {$w_2^-$};
        \node[world] (w4p) at (4.5, 0) {$w_3^+$};
        \node[world] (w4m) at (4.5,-1) {$w_3^-$};

        % \draw[info, green!50!black] (-.6,.6) node[above]{$p_1$} rectangle (.6,-.6);
        % \draw[info, blue!100!white] (-.5,.5) rectangle ( .5,-1.5) node[below]{$q_1$};
        % \draw[info, green!65!black] ( .9,.6) node[above]{$p_2$} rectangle (2.1,-.6);
        % \draw[info,  blue!80!white] (  1,.5) rectangle (  2,-1.5) node[below]{$q_2$};
        % \draw[info, green!80!black] (3.9,.6) node[above]{$p_3$} rectangle (5.1,-.6);
        % \draw[info,  blue!60!white] (  4,.5) rectangle (  5,-1.5) node[below]{$q_3$};

    \end{tikzpicture}
    \caption{The two maximal states satisfying $D_1$ in the model $\SM_4$.}
    \label{fig:Dformula}
\end{figure}
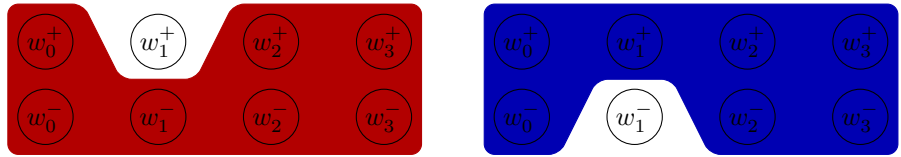

\begin{align}
    S_{0} &:= \bigivee_{i=0}^{l-1} \left( \neg C_i^+ \ivee \neg C_i^-  \right)  \\
    S_k &:= \bigivee_{i=0}^{k-1} \left( \neg C_i^+ \land \neg C_i^- \right) \;\ivee\; \bigivee_{i=k}^{l-1} \left( \neg C_i^+ \ivee \neg C_i^- \right)  \quad\text{for $0 < k < l$}  \\
    S_l &:= \bigivee_{i=0}^{l-1} \left( \neg C_i^+ \land \neg C_i^- \right)
\end{align}
The semantics of the formulas $S_k$ is more complex to describe, but we omit a thorough analysis since we only need the following technical lemma.

\begin{lemma}\label{lemma:propertyOfS}
    Let $s$ be a $k$-switching of the model $\SM_l$.
    Then $\SM_l, s \nvDash S_k$, but for every proper subset $t \subsetneq s$ it holds $\SM_l, t \vDash S_k$.
\end{lemma}
\begin{proof}
    Firstly, let us show that $s \nvDash S_k$ (for ease of read, in the rest of the proof we will omit the model $\model{S}_l$).
    Since $S_k$ is a disjunction, we just need to check that each disjunct is not satisfied at $s$.
    \begin{itemize}
        \item For $i < k$, by Definition \ref{definition:switching} $s$ contains exactly one world among $w_i^+$ and $w_i^-$.
        This means that $s \vDash C_i^+$ or $s\vDash C_i^-$.
        In either case we have $s \nvDash \neg C_i^+ \land \neg C_i^-$.

        \item For $i \geq k$, by Definition \ref{definition:switching} $s$ contains both worlds $w_i^+$ and $w_i^-$.
        This means that $s \nvDash \neg C_i^+$ (since it contains $w_i^+$) and $s \nvDash \neg C_i^-$ (since it contains $w_i^-$).
        So we have $s \nvDash \neg C_i^+ \ivee \neg C_i^-$.
    \end{itemize}

    \noindent
    Since these are all the disjuncts of $S_k$, the formula is not satisfied at $s$.

    Secondly, let us show that every $t \subsetneq s$ satisfies $S_k$.
    Since $S_k$ is a disjunction, we just need to find a disjunct satisfied by $t$.
    And since the containment $t \subsetneq s$ is strict, there exists a world in $w_i^e$ in $s$ and not in $t$.
    We consider two cases, depending whether $i < k$ or $i\geq k$.
    \begin{itemize}
        \item If $i < k$, by Definition \ref{definition:switching} $t$ does not contain $w_i^+$ nor $w_i^-$ (one of them was already missing from $s$).
        So $t \vDash \neg C_i^+ \ivee \neg C_i^-$, which is a disjunct of $S_k$, thus $t \vDash S_k$.

        \item If $i \geq k$, by Definition \ref{definition:switching} $t$ contains exactly one among $w_i^+$ and $w_i^-$ ($s$ being a $k$-switching contains both of them).
        This means that either $t \vDash \neg C_i^+$ (if $w_i^+ \notin t$) or $t \vDash \neg C_i^-$ (if $w_i^- \notin t$).
        In both cases we have $t \vDash \neg C_i^+ \ivee \neg C_i^-$, which is a disjunt of $S_k$, thus $t \vDash S_k$.
    \end{itemize}

    \noindent
    So it follows that $t \vDash S_k$, as desired.
\end{proof}

% \medskip
% \noindent
% \textbf{The translations.}
\subsection{The translation}
\label{subsection:translation}
We have all the elements to define our encoding of an instance of $\TQBF$ into an instance of $\MC(\inqB)$, that is, an information model, a state and an \emph{inquisitive formula}.
We firstly focus on the most complex part:
defining the inquisitive formula.

For ease of read, we use the symbols $\zeta,\xi,\eta$ to indicate propositional formulas (without quantifiers, so that an arbitrary QBF is of the form $Q_0 x_0 \dots Q_{l-1} x_{l-1} \zeta$) and the symbols $\phi,\psi,\chi$ to indicate inquisitive formulas.
Moreover, we assume that propositional formulas contain only propositional variables from the set $\APset_l = \{x_0,\dots,x_{l-1}\}$.
Finally, we assume that all propositional formulas are in \emph{negation normal form}, that is, for negations only appear in front of atomic propositions.\footnote{Translating a formula in negation normal form is an operation linear in time, thus this assumption does not affect our complexity analysis.}

We give the translation in two steps:
firstly, we define by recursion \emph{two} translations of \emph{propositional formulas} into inquisitive formulas with an associated polarity, a \emph{positive translation} $\zeta^p$ and a \emph{negative translation} $\zeta^n$.
Then we provide, again by mutual recursion, two translations for quantified Boolean formulas with an associated polarity, $( Q_0 x_0 \dots Q_{l-1} x_{l-1}  \zeta )^P$ and $( Q_0 x_0 \dots Q_{l-1} x_{l-1}  \zeta )^N$.

The translations $\zeta^p$ and $\zeta^n$ are defined by the following clauses:
\begin{equation*}
\begin{array}{l@{\hspace{.2em}}l @{\hspace{4em}} l@{\hspace{.2em}}l}
    (x_i)^p &:= q_i \to p_i
        &(x_i)^n &:= q_i \to \neg p_i  \\
    (\neg x_i)^p &:= q_i \to \neg p_i
        &(\neg x_i)^n &:= q_i \to p_i  \\
    (\zeta \land \xi)^p &:= \zeta^p \land \xi^p
        &(\zeta \land \xi)^n &:= \zeta^n \ivee \xi^n  \\
    (\zeta \vee \xi)^p &:= \zeta^p \ivee \xi^p
        &(\zeta \vee \xi)^n &:= \zeta^n \land \xi^n
\end{array}
\end{equation*}

\noindent
This translation allows us to encode the semantics of $\zeta$ in terms of switching states instead of Boolean valuations.
In particular, the formula $\zeta$ is satisfied under a certain Boolean valuation $\sigma$ iff $\zeta^p$ is supported by the corresponding switching $s_\sigma$ (recall Definition \ref{definition:switching}).

\begin{lemma}\label{lemma:propositionalTranslation}
    Let $\sigma: \APset_l \to \{0,1\}$ be a Boolean valuation and let $s_\sigma$ be the corresponding switching of $\SM_l$.
    Then for every propositional formula $\zeta$ we have that:
    \begin{equation*}
    \begin{array}{l @{\qquad\text{iff}\qquad} l @{\qquad\text{iff}\qquad} l}
        \sigma(\zeta) = 1
            &s_\sigma \vDash \zeta^p
            &s_\sigma \nvDash \zeta^n \\
        \sigma(\zeta) = 0
            &s_\sigma \vDash \zeta^n
            &s_\sigma \nvDash \zeta^p
    \end{array}
    \end{equation*}
\end{lemma}
\begin{proof}
    We proceed by induction on the structure of the formula $\zeta$.
    We firstly prove that $\sigma(\zeta) = 1 \iff s_\sigma \vDash \zeta^p$ and $\sigma(\zeta) = 0 \iff s_\sigma \vDash \zeta^n$.
    \begin{itemize}
        \item If $\zeta = x_i$, we have
        % \begin{gather*}
        %     \sigma(x_i) = \top  \;\iff\;  s_\sigma \vDash r_i \to p_i  \;\iff\;  s_\sigma \vDash (x_i)^p  \\
        %     \sigma(x_i) = \bot  \;\iff\;  s_\sigma \vDash r_i \to \neg p_i  \;\iff\;  s_\sigma \vDash (x_i)^n
        % \end{gather*}
        \begin{equation*}
        \begin{array}{l@{\;\iff\;}l @{\hspace{4em}} l@{\;\iff\;}l}
            \sigma(x_i) = 1  & s_\sigma \vDash r_i \to p_i
                &\sigma(x_i) = 0  &  s_\sigma \vDash r_i \to \neg p_i  \\
            &s_\sigma \vDash (x_i)^p
                &&  s_\sigma \vDash (x_i)^n
        \end{array}
        \end{equation*}

        \item If $\zeta = \neg x_i$, we have
        % \begin{gather*}
        %     \sigma(\neg x_i) = \top  \;\iff\; \sigma(x_i) = \bot  \;\iff\;  s_\sigma \vDash r_i \to \neg p_i  \;\iff\;  s_\sigma \vDash (\neg x_i)^p  \\
        %     \sigma(\neg x_i) = \bot  \;\iff\;  \sigma(x_i) = \top  \;\iff\;  s_\sigma \vDash r_i \to p_i  \;\iff\;  s_\sigma \vDash (\neg x_i)^n
        % \end{gather*}
        \begin{equation*}
        \begin{array}{l@{\;\iff\;}l @{\hspace{4em}} l@{\;\iff\;}l}
            \sigma(\neg x_i) = 1  &\sigma(x_i) = 0
                &\sigma(\neg x_i) = 0  &\sigma(x_i) = 1  \\
            &s_\sigma \vDash r_i \to \neg p_i
                &&s_\sigma \vDash r_i \to p_i  \\
            &s_\sigma \vDash (\neg x_i)^p
                &&s_\sigma \vDash (\neg x_i)^n
        \end{array}
        \end{equation*}

        \item If $\zeta = \eta \land \xi$, we have
        % \begin{gather*}
        %     \sigma(\eta \land \xi) = \top  \;\iff\;  \sigma(\eta) = \sigma(\xi) = \top  \;\iff\;  s_\sigma \vDash \eta^p \land \xi^p  \;\iff\;  s_\sigma \vDash (\eta \land \xi)^p  \\
        %     \sigma(\eta \land \xi) = \bot  \;\iff\;  \sigma(\eta) = \bot \;\text{or}\; \sigma(\xi) = \bot  \;\iff\;  s_\sigma \vDash \eta^n \;\text{or}\; s_\sigma \vDash \xi^n  \;\iff\;  s_\sigma \vDash (\eta \land \xi)^n
        % \end{gather*}
        \begin{equation*}
        \begin{array}{l@{}l @{\hspace{4em}} l@{}l}
            &\sigma(\eta \land \xi) = 1
                &&  \sigma(\eta \land \xi) = 0  \\
            &\;\iff\;  \sigma(\eta) = \sigma(\xi) = 1
                &&\;\iff\;  \sigma(\eta) = 0 \;\text{or}\; \sigma(\xi) = 0  \\
            &\;\iff\;  s_\sigma \vDash \eta^p \land \xi^p
                &&\;\iff\;  s_\sigma \vDash \eta^n \;\text{or}\; s_\sigma \vDash \xi^n  \\
            &\;\iff\;  s_\sigma \vDash (\eta \land \xi)^p
                &&\;\iff\;  s_\sigma \vDash (\eta \land \xi)^n
        \end{array}
        \end{equation*}

        \item If $\zeta = \eta \vee \xi$, we have
        % \begin{gather*}
        %     \sigma(\eta \vee \xi) = \top  \;\iff\;  \sigma(\eta) = \top \;\text{or}\; \sigma(\xi) = \top  \;\iff\;  s_\sigma \vDash \eta^p \ivee \xi^p  \;\iff\;  s_\sigma \vDash (\eta \vee \xi)^p  \\
        %     \sigma(\eta \vee \xi) = \bot  \;\iff\;  \sigma(\eta) = \sigma(\xi) = \bot  \;\iff\;  s_\sigma \vDash \eta^n \land \xi^n  \;\iff\;  s_\sigma \vDash (\eta \vee \xi)^n
        % \end{gather*}
        \begin{equation*}
        \begin{array}{l@{}l @{\hspace{4em}} l@{}l}
            &\sigma(\eta \vee \xi) = 1
                &&  \sigma(\eta \vee \xi) = 0  \\
            &\;\iff\;  \sigma(\eta) = 1 \;\text{or}\; \sigma(\xi) = 1
                &&\;\iff\;  \sigma(\eta) = \sigma(\xi) = 0  \\
            &\;\iff\;  s_\sigma \vDash \eta^p \ivee \xi^p
                &&\;\iff\;  s_\sigma \vDash \eta^n \land \xi^n  \\
            &\;\iff\;  s_\sigma \vDash (\eta \vee \xi)^p
                &&\;\iff\;  s_\sigma \vDash (\eta \vee \xi)^n
        \end{array}
        \end{equation*}
    \end{itemize}

    Secondly, we notice that:
    \begin{gather*}
        s_\sigma \vDash \zeta^p  \;\iff\;  \sigma(\zeta) = 1  \;\iff\;  \sigma(\zeta) \ne 0  \;\iff\;  s_\sigma \nvDash \zeta^n  \\
        s_\sigma \vDash \zeta^n  \;\iff\;  \sigma(\zeta) = 0  \;\iff\;  \sigma(\zeta) \ne 1  \;\iff\;  s_\sigma \nvDash \zeta^p
    \end{gather*}

\noindent
This concludes the proof.

\end{proof}

% \gnote{Add example?}

\noindent
Now we shift our attention to quantified Boolean formulas, and in particular to the formulas $\theta_k := Q_k x_k \dots Q_{l-1} x_{l-1} \zeta$.
Notice that the lower the value $k$, the more quantifiers appear in the formula $\theta_k$, with the limit case being $\theta_l := \zeta$.

As for propositional formulas we define two translations by mutual recursion.
The two translations are indicated with $\theta_k^P$ and $\theta_k^N$.
These translations involve the formulas $D_i$ and $S_i$ introduced in Subsection \ref{subsection:specialFormulas}.

\begin{equation*}
\begin{array}{lll}
    \zeta^P &= &\zeta^p  \\
    \zeta^N &= &\zeta^n  \\[.2em]
    \theta_{k-1}^P  &=  &\left\{\begin{array}{ll}
        D_{k-1} \to \theta_k^P
            &\text{if}\; Q_{k-1} = \forall  \\
        \left( D_{k-1} \to \theta_k^N \right) \to S_k
            &\text{if}\; Q_{k-1} = \exists
    \end{array}\right.  \\[1em]
    \theta_{k-1}^N  &=  &\left\{\begin{array}{ll}
        (D_{k-1} \to \theta_k^P ) \to S_k
            &\text{if}\; Q_{k-1} = \forall  \\
        D_{k-1} \to \theta_k^N
            &\text{if}\; Q_{k-1} = \exists
    \end{array}\right.
\end{array}
\end{equation*}

\noindent
Similarly to the translations $\zeta^p$ and $\zeta^n$, we can prove a preservation result for these translations.

\begin{lemma}\label{lemma:translationWorks}
    Consider the quantified Boolean formula $\theta_k = Q_{k}x_{k}\dots Q_{l-1} x_{l-1} \zeta$ with variables in $\APset_l$.
    Let $\sigma: \APset_{k} \to \{0,1\}$ be a Boolean valuation and $s_\sigma$ the corresponding $k$-switching of $\SM_l$.
    We have that:
    \begin{equation}\label{eq:translation}
    \begin{array}{l c l}
        \sigma(\theta_k) = 1
            &\text{iff}
            &\SM_l, s_\sigma \vDash (\theta_k)^P  \\
        \sigma(\theta_k) = 0
            &\text{iff}
            &\SM_l, s_\sigma \vDash (\theta_k)^N
    \end{array}
    \end{equation}
\end{lemma}
\begin{proof}
    We prove the result by induction on the number of quantifiers preceding $\zeta$.
    The base case is
    \begin{equation*}
    \begin{array}{l c l}
        \sigma(\zeta) = 1
            &\text{iff}
            &\SM_l, s_\sigma \vDash (\zeta)^P  \\
        \sigma(\zeta) = 0
            &\text{iff}
            &\SM_l, s_\sigma \vDash (\zeta)^N
    \end{array}
    \end{equation*}
    and this is the statement of Lemma \ref{lemma:propositionalTranslation}.

    As for the inductive step, suppose the equivalences to hold for formulas with $l-k-1$ quantifiers.
    We want to show that the equivalences hold also for formulas with $l-k$ quantifiers.
    So consider the formula $\theta_{k-1} = Q_{k-1} x_{k-1} \dots Q_{l-1} x_{l-1} \zeta$.
    As a useful shorthand, we will write $Q_{k-1} x_{k-1} \theta_k$ for the formula $\theta_{k-1}$, making explicit reference to the relation between the two formulas.

    Let $\sigma : \APset_{k-1} \to \{0,1\}$ and $s_\sigma$ be as in the statement of the lemma.
    We have two cases two consider, depending on whether $Q_{k-1} = \forall$ or $Q_{k-1} = \exists$.
    \begin{itemize}
        \item Suppose $Q_{k-1} = \forall$ and let us firstly focus on the first equivalence of \ref{eq:translation}.
        We have:
        \begin{equation*}
        \begin{array}{ll}
            &\SM_l, s_\sigma \vDash (\forall x_{k-1} \theta_k)^P  \\
            \iff  &\SM_l, s_\sigma \vDash D_{k-1} \to \theta_k^P  \\
            \iff  &\text{For every $t\subseteq s_\sigma$, if } \SM_l, t \vDash D_{k-1} \text{ then } \SM_l, t \vDash \theta_k^P
        \end{array}
        \end{equation*}

        By persistency of the logic (Proposition \ref{proposition:empty and persistency}), we just need to check the condition for the \emph{maximal} substates of $s_\sigma$ satisfying $D_{k-1}$.
        By Lemma \ref{lemma:dFormula} these maximal substates are $t^+ = s_\sigma \setminus\{ w_{k-1}^- \}$ and $t^- = s_\sigma \setminus \{ w_{k-1}^+ \}$, which are both $k$-switchings and substates of $s_\sigma$.
        These switchings correspond respectively to the Boolean valuations $\tau^+, \tau^-: \APset_k \to \{0,1\}$ defined as
        \begin{equation*}
        \begin{array}{ll}
            \tau^+( x_i ) = \left\{  \begin{array}{ll}
                \sigma(x_i)     &\text{if } i<k  \\
                1            &\text{if } i=k
            \end{array}  \right.  \\[1em]
            \tau^-( x_i ) = \left\{  \begin{array}{ll}
                \sigma(x_i)     &\text{if } i<k  \\
                0            &\text{if } i=k
            \end{array}  \right.
        \end{array}
        \end{equation*}

        By inductive hypothesis, we have
        \begin{equation*}
        \begin{array}{l l l}
            \SM_l, t^+ \vDash \theta_k^P
                &\iff
                &\tau^+( \theta_k ) = 1  \\
            \SM_l, t^- \vDash \theta_k^P
                &\iff
                &\tau^-( \theta_k ) = 1
        \end{array}
        \end{equation*}
        So in conclusion we have
        \begin{equation*}
        \begin{array}{ll}
            &\SM_l, s_\sigma \vDash (\forall x_{k-1} \theta_k )^P   \\
            \iff   &\tau^+( \theta_k ) = \tau^-( \theta_k ) = 1   \\
            \iff   &\sigma( \forall x_{k-1} \theta_k ) = \sigma( \theta_{k-1} ) = 1
        \end{array}
        \end{equation*}

        As for the second equivalence of \ref{eq:translation}, we have:
        \begin{equation*}
        \begin{array}{ll}
            &\SM_l, s_\sigma \vDash (\forall x_{k-1} \theta_k)^N  \\
            \iff  &\SM_l, s_\sigma \vDash ( D_{k-1} \to \theta_k^P ) \to S_k  \\
            \iff  &\text{For every $t\subseteq s_\sigma$, if } \SM_l, t \vDash D_{k-1} \to \theta_k^P \text{ then } \SM_l, t \vDash S_k
        \end{array}
        \end{equation*}
        By Lemma \ref{lemma:propertyOfS}, the only substate of $s_\sigma$ not satisfying $S_k$ is $s_\sigma$ itself.
        So the chain of equivalences above can be extended as follows:
        \begin{equation*}
        \begin{array}{ll}
            &\SM_l, s_\sigma \vDash (\forall x_{k-1} \theta_k)^N  \\
            \iff  &\SM_l, s_\sigma \nvDash D_{k-1} \to \theta_k^P  \\
            \iff  &\text{For some $t\subseteq s_\sigma$, } \SM_l, t \vDash D_{k-1} \text{ and } \SM_l, t \nvDash \theta_k^P  \\
            \stackrel{\ast}{\iff}  &\SM_l, t^+ \nvDash \theta_k^P \text{ or } \SM_l, t^- \nvDash \theta_k^P  \\
            \stackrel{\ast\ast}{\iff} &\tau^+(\theta_k) = 0 \text{ or } \tau^-(\theta_k) = 0  \\
            \iff &\sigma(\forall x_{k-1} \theta_k) = \sigma(\theta_{k-1}) = 0
        \end{array}
        \end{equation*}
        where the equivalence $(\ast)$ follows by persistency and the equivalence $(\ast\ast)$ follows by two applications of the inductive hypothesis applied to the $k$-switchings $t^+$ and $t^-$ respectively.

        \item Suppose $Q_k = \exists$.
        Most of the arguments and passages follow closely the structure of the previous case, so we will omit comments to most of them.

        Let us firstly focus on the first equivalence of \ref{eq:translation}.
        We have:
        \begin{equation*}
        \begin{array}{ll}
            &\SM_l, s_\sigma \vDash (\exists x_{k-1} \theta_k)^P  \\
            \iff  &\SM_l, s_\sigma \vDash ( D_{k-1} \to \theta_k^N ) \to S_k  \\
            \iff  &\text{For every $t\subseteq s_\sigma$, if } \SM_l, t \vDash D_{k-1} \to \theta_k^N \text{ then } \SM_l, t \vDash S_k  \\
            \iff  &\SM_l, s_\sigma \nvDash D_{k-1} \to \theta_k^N  \\
            \iff  &\text{For some $t\subseteq s_\sigma$, } \SM_l, t \vDash D_{k-1} \text{ and } \SM_l, t \nvDash \theta_k^N  \\
            \iff  &\SM_l, t^+ \nvDash \theta_k^N \text{ or } \SM_l, t^- \nvDash \theta_k^N  \\
            \stackrel{\ast}{\iff} &\tau^+(\theta_k) = 1 \text{ or } \tau^-(\theta_k) = 1  \\
            \iff &\sigma(\exists x_{k-1} \theta_k) = \sigma(\theta_{k-1}) = 1
        \end{array}
        \end{equation*}
        where the equivalence $(\ast)$ follows by inductive hypothesis.

        As for the second equivalence of \ref{eq:translation}, we have:
        \begin{equation*}
        \begin{array}{ll}
            &\SM_l, s_\sigma \vDash (\exists x_{k-1} \theta_k)^N  \\
            \iff  &\SM_l, s_\sigma \vDash D_{k-1} \to \theta_k^N  \\
            \iff  &\text{For every $t\subseteq s_\sigma$, if } \SM_l, t \vDash D_{k-1} \text{ then } \SM_l, t \vDash \theta_k^N  \\
            \iff  &\SM_l, t^+ \vDash \theta_k^N \text{ and } \SM_l, t^- \vDash \theta_k^N  \\
            \stackrel{\ast}{\iff} &\tau^+( \theta_k ) = \tau^-( \theta_k ) = 0  \\
            \iff &\sigma( \exists x_{k-1} \theta_k ) = \sigma( \theta_{k-1} ) = 0
        \end{array}
        \end{equation*}
        where the equivalence $(\ast)$ follows from inductive hypothesis.

    \end{itemize}

    \noindent
    And this concludes the proof of the inductive step.
    Thus the result follows by induction.

\end{proof}

\noindent
As a direct corollary of the previous lemma we obtain the desired result:
we can reduce an instance of $\TQBF$ to an instance of $\MC(\inqB)$.
\begin{corollary}\label{corollary:reduction}
    Let $\theta := Q_0 x_0 \dots Q_{l-1} x_{l-1} \zeta$ be a closed quantified Boolean formula.
    Then $\theta$ is valid iff $\SM_l, W_l \vDash \theta^P$.
    In other terms:
    \begin{equation*}
        \theta \in \TQBF \quad\text{iff}\quad \tuple{\SM_l,W_l,\theta^P} \in \MC(\inqB)
    \end{equation*}
\end{corollary}
\begin{proof}
    Recall that the set $W_l$ of all the worlds of $\SM_l$ is a $0$-switching, and that the corresponding valuation is the empty function.
    By Lemma \ref{lemma:translationWorks} (applied for $k=0$) we have that $\theta$ is valid (i.e., satisfied by the empty function) iff $\SM_l,W_l \vDash \theta^P$.
\end{proof}

% \medskip
% \noindent
% \textbf{Complexity.}
\subsection{Complexity}
It remains to show that the reduction of $\TQBF$ to $\MC(\inqB)$ presented in Corollary \ref{corollary:reduction} is polynomially bounded in size.
This entails that $\MC(\inqB)$ is a $\PSPACE$-hard problem, since so is $\TQBF$.

Let $\theta$ be the formula $Q_0 x_0 \dots Q_{l-1} x_{l-1} \zeta$ with $\size{\zeta} = h$, and thus $\size{\theta} = \bigTheta(l + h)$.
Using the encoding of information models presented in Section \ref{section:algorithm}, the size of $\SM_l$ is $\card{\{p_0,\dots,p_{l-1},q_0,\dots,q_{l-1}\}} \cdot \card{W_l} =  2l^2 $, which is indeed polynomially bounded by $l+h$.
And the same goes for the encoding of the information state $W_l$, which is of size $2l$.

As for the size of $\theta^P$, the computations are slightly more involved, so we proceed one step at a time.
Since there are several summations and an inductive proof involved, the $\bigO$ notation may easily lead to errors.
For this reason, we firstly present a semi-formal analysis of the complexity involving $\bigO$s in the current section, and we give a more thorough analysis in Appendix \ref{section:explicitAnalysis}.

Firstly, we find bounds for the size of the special formulas presented in Subsection \ref{subsection:specialFormulas}.
Recall that we indicate with $\size{\phi}$ the size of the encoding of a formula $\phi$.
\begin{itemize}
    \item $\size{ C^+_k } = \size{\;  q_k \land p_k  \;} = \bigO(\log(l))$ and $\size{ C^-_k } = \size{\;  q_k \land \neg p_k  \;} = \bigO(\log(l))$.
    % \footnote{We highlight a subtle point about the use of the $\bigO$ notation here and in the following passages. We are accounting also for the parameter $k$ in computing the bound in size. So that, for example, $\size{ C^+_k }$ is a function bounded by $cl$ for some constant $c$ not dependent on $l$ nor on $k$.}

    \item $\size{ \neg C_k^+ \land \neg C_k^- } = \size{ \neg C_k^+ \ivee \neg C_k^- } = \bigO(\log(l))$.

    \item $\size{ S_0 }  = \bigO(l) + \sum_{i=0}^{l-1} \size{\neg C^+_i \ivee \neg C^-_i }  = \bigO(l) + \bigO(l\log(l) ) = \bigO(l\log(l))$.
    % \footnote{Connected to the previous footnote: the second identity holds since the constant bounding $\frac{ \size{\neg C^+_i \ivee \neg C^-_i } }{\log(l)}$ does not depend on the value $i$.}

    \item For $k > 0$ we have $\size{ S_k } = \bigO(l) + \sum_{i=0}^{k-1} \size{\neg C_i^+ \land \neg C_i^-} + \sum_{i=k}^{l-1} \size{\neg C_i^+ \ivee \neg C_i^-} = \bigO(l) + \bigO(k \log(l)) + \bigO((l-k)\log(l)) = \bigO(l\log(l))$.

    \item $\size{ D_k } = \size{\;  q_k \to ( p_k \ivee \neg p_k)  \;} = \bigO(\log(l))$.

    \item A straightforward inductive argument on the construction of $\zeta^P$ (resp., $\zeta^N$) shows that $\size{\zeta^P} = \bigO(h)$ (resp., $\size{\zeta^N} = \bigO(h)$).
\end{itemize}

\noindent
And now we prove by (descending) induction on $k$ that $\size{ \theta_k^P }$ and $\size{ \theta_k^N }$ are both in $\bigO((l-k-1) l \log(l) + h)$.
The base step of the induction $k=l$ (i.e., when there are no quantifiers involved and the formulas are $\zeta^P$ and $\zeta^N$ respectively) was already considered.
So assume that $\size{ \theta_k^P }$ and $\size{ \theta_k^N }$ are both in $\bigO((l-k-1) l \log(l) + h)$.
We prove only that $\size{ \theta_{k-1}^P }$ is in $\bigO((l-k) l \log(l) + h)$, since the other inequality follows similar computations.
We consider two separate cases, depending on whether $Q_{k-1} = \forall$ or $Q_{k-1} = \exists$.
\begin{equation*}
\begin{array}{ll}
    \size{ (\forall x_{k-1} \theta_k)^P }
    &= \size{ D_{k-1} \to \theta_k^P } \\
    &= \bigO(\log(l)) + \bigO((l-k-1) l \log(l) + h) \\
    &= \bigO((l-k) l \log(l) + h)
\end{array}
\end{equation*}
\begin{equation*}
\begin{array}{ll}
    \size{ (\exists x_{k-1} \theta_k)^P }
    &= \size{ ( D_{k-1} \to \theta_k^N ) \to S_k } \\
    &= \bigO(l \log(l)) + \bigO((l-k-1) l \log(l) + h) + \bigO(l\log(l)) \\
    &= \bigO((l-k) l \log(l) + h)
\end{array}
\end{equation*}

\noindent
And this concludes the inductive proof.
The salient results we thus obtain are the following.

\begin{theorem}\label{theorem:boundSizeFormula}
    For $\theta = Q_0 x_0 \dots Q_{l-1} x_{l-1} \zeta$ a closed quantified Boolean formula, the size of $\theta^P$ is in $\bigO(  l^2 \log(l) + \size{\zeta}  )$.
    In particular, the size of the translation is polynomially bounded by the size of the original formula.
\end{theorem}

\begin{corollary}\label{corollary:mc inqb is pspace hard}
    The problem $\MC(\inqB)$ is $\PSPACE$-hard.
\end{corollary}

\noindent
As anticipated, this was the last ingredient needed to show that the four problems considered are $\PSPACE$-complete.

\begin{theorem}
    The problems $\MC(\inqB)$, $\MCanti(\inqB)$, $\MC(\inqM)$ and $\MCanti(\inqM)$ are $\PSPACE$-complete.
\end{theorem}
% \begin{proof}
%     By Corollary \ref{corollary:mc inqb is pspace hard}, we have that $\MC(\inqB)$ is a $\PSPACE$-hard problem, and also $\MCanti(\inqB)$ is $\PSPACE$-hard since it reduces to $\MC(\inqB)$.
%     By Theorem \ref{theorem:mc inqm and mcanti inqm are pspace}, the problems $\MC(\inqM)$ and $\MCanti(\inqM)$ are both in $\PSPACE$.
%     Since $\MC(\inqB)$ and $\MCanti(\inqB)$ are simpler than 
% \end{proof}

\section{Conclusions}\label{section:conclusions}

%% Recap
In this paper we formalized and studied the model checking problems $\MC(\inqB)$ and $\MC(\inqM)$ for inquisitive propositional logic and inquisitive modal logic, and studied their complexity.
In Section \ref{section:algorithm} we provided an algorithm to solve the problem $\MC(\inqM)$ and proved its correctness.
This allowed us to show that the complexity of the problem is in the class $\PSPACE$.
In Section \ref{section:complexity} we provided a polynomial-space reduction of the \emph{true quantified Boolean formula} problem to $\MC(\inqB)$, thus showing that the problem is also $\PSPACE$-hard.
Putting these results together, we inferred that both problems are $\PSPACE$-complete.

%% Further questions
This paper is the first investigation of computational complexity issues for inquisitive logics.
A venue for further investigation is whether the proof presented adapts to other inquisitive logics, such as \emph{inquisitive epistemic logic} \cite[Ch.~7]{Ciardelli:16}, \emph{inquisitive intuitionistic logic} \cite{Ciardelli:20inqi} or first-order versions of the system such as $\inqBQ$ \cite[Ch.~4]{Ciardelli:16}.

\bibliographystyle{plain}
\bibliography{complexity_mc}

\appendix
\section{Explicit analysis}\label{section:explicitAnalysis}

\noindent
In this appendix we give an explicit analysis of the computations presented at the end of Section \ref{section:complexity}, leading to Theorem \ref{theorem:boundSizeFormula}.
Throughout this appendix, we indicate with $c_i$ constants dependent only on the encoding used for formulas, chosen opportunely to satisfy the inequalities involved.

Firstly, we give an explicit analysis of the sizes of the formulas presented in Subsection \ref{subsection:specialFormulas}.
\begin{itemize}
    \item $\size{ C^+_k } = \size{\;  q_k \land p_k  \;} \leq c_1 \log(l)$ and $\size{ C^-_k } = \size{\;  q_k \land \neg p_k  \;} \leq c_1 \log(l)$.

    \item $\size{ \neg C_i^+ \land \neg C_i^- } = \size{ \neg C_i^+ \ivee \neg C_i^- } \leq c_2 + \size{ C_i^+ } + \size{C_i^-} \leq c_3 \log(l)$.

    \item $\size{ S_0 }  \leq c_4 l + \sum_{i=0}^{l-1} \size{\neg C^+_i \ivee \neg C^-_i }  \leq c_4 l + l \cdot c_2\log(l) \leq c_5 l \log(l)$.

    \item For $k > 0$ we have $\size{ S_k } = c_6 l + \sum_{i=0}^{k-1} \size{\neg C_i^+ \land \neg C_i^-} + \sum_{i=k}^{l-1} \size{\neg C_i^+ \ivee \neg C_i^-} \leq c_6 l + k \cdot c_3 \log(l) + (l-k) \cdot c_3 \log(l) \leq c_7 l \log(l)$.

    \item To make computations less tedious, we define $c_8 = \max(c_5,c_7)$, so that $\size{S_k} \leq c_8 l \log(l)$ independently from the value of $k$.

    \item $\size{ D_k } = \size{\;  q_k \to ( p_k \ivee \neg p_k)  \;} \leq c_9 \log(l)$.

    \item A straightforward inductive argument on the construction of $\zeta^P$ (resp., $\zeta^N$) shows that $\size{\zeta^P} = \bigO(h)$ (resp., $\size{\zeta^N} = \bigO(h)$).
\end{itemize}

\noindent
Secondly, we provide an explicit analysis of the size of $\size{ (Q_k x_k \dots Q_{l-1} x_{l-1} \zeta)^P }$ and $\size{ (Q_k x_k \dots Q_{l-1} x_{l-1} \zeta)^N }$.
As in Subsection \ref{subsection:translation} we use the shorthand $\theta_k := Q_k x_k \dots Q_{l-1} x_{l-1} \zeta$, and indicate with $Q_{k-1} x_{k-1} \theta_k$ the formula $\theta_{k-1}$ to make the relation between the formulas $\theta_k$ and $\theta_{k-1}$ explicit.
We show by downward induction on $k$ that both quantities are bounded by $c ( (l-k-1) l \log(l) + h)$ for an opportune choice of the constant $c$.
The base step of the induction ($k=l$, i.e., when no quantifiers are involved and the the formulas considered are $\zeta^P$ and $\zeta^N$) has already been taken care of.
So assume that $\size{ \theta_k^P } \leq c ( (l-k) l \log(l) + h)$ and $\size{ \theta_k^N } \leq c ( (l-k) l \log(l) + h)$.

As for the inductive step, we firstly prove the inequality $\size{ \theta_{k-1}^P } \leq c ( (l-(k-1)) l \log(l) + h)$.
We consider two separate cases, depending on whether $Q_{k-1} = \forall$ or $Q_{k-1} = \exists$.
\begin{equation*}
\begin{array}{lll}
    \size{ (\forall x_{k-1} \theta_k)^P }
    &= \size{ D_{k-1} \to \theta_k^P } \\
    &\leq c_{10} \log(l) + c ( (l-k) l \log(l) + h) \\
    &\leq c \log(l) + c(l-k)l \log(l) + c h
        &\text{(if $c \geq c_{10}$)}  \\
    &\leq c (\; (l-(k-1)) l \log(l) + h \;)
\end{array}
\end{equation*}
\begin{equation*}
\begin{array}{lll}
    \size{ (\exists x_{k-1} \theta_k)^P }
    &= \size{ ( D_{k-1} \to \theta_k^N ) \to S_k } \\
    &\leq c_{11} l \log(l) + c ( (l-k) l \log(l) + h) \\
    &\leq c l \log(l) + c(l-k)l \log(l) + c h
        &\text{(if $c \geq c_{11}$)}  \\
    &\leq c (\; (l-(k-1)) l \log(l) + h \;)
\end{array}
\end{equation*}

\noindent
Secondly we prove the inequality $\size{ \theta_{k-1}^N } \leq c ( (l-(k-1)) l \log(l) + h)$.
Also in this case we distinguish the cases $Q_{k-1} = \forall$ and $Q_{k-1} = \exists$.
\begin{equation*}
\begin{array}{lll}
    \size{ (\forall x_{k-1} \theta_k)^N }
    &= \size{ ( D_k \to \theta_k^P ) \to S_k } \\
    &\leq c_{12} l \log(l) + c ( (l-k) l \log(l) + h) \\
    &\leq c l \log(l) + c(l-k)l \log(l) + c h
        &\text{(if $c \geq c_{12}$)}  \\
    &\leq c (\; (l-(k-1)) l \log(l) + h \;)
\end{array}
\end{equation*}
\begin{equation*}
\begin{array}{lll}
    \size{ (\exists x_{k-1} \theta_k)^N }
    &= \size{ D_k \to \theta_k^N } \\
    &\leq c_{13} \log(l) + c ( (l-k) l \log(l) + h) \\
    &\leq c \log(l) + c(l-k)l \log(l) + c h
        &\text{(if $c \geq c_{13}$)}  \\
    &\leq c (\; (l-(k-1)) l \log(l) + h \;)
\end{array}
\end{equation*}

\noindent
So by choosing $c := \max(c_{10}, c_{11}, c_{12}, c_{13})$ we meet all the conditions necessary for the inequality to hold.
This concludes our inductive proof, and shows that both $\size{ \theta_k^P }$ and $\size{ \theta_k^N }$ are in $\bigO(( (l-k) l \log(l) + h))$, as desired.

\end{document}